\documentclass[11pt,draftcls,onecolumn]{IEEEtranwc}

\usepackage{color}
\usepackage{amsmath,amssymb,amscd}
\usepackage{subfigure, epsfig}
\newcommand{\mb}{\mathbf}
\newcommand{\ul}{\underline}
\newcommand{\ol}{\overline}

\newcommand{\mc}{\mathcal}

\newcommand{\ds}{\displaystyle}

\newcommand{\mr}{\mathrm}

\newtheorem{theorem}{Theorem}[section]

\newtheorem{conjecture}[theorem]{Conjecture}
\newtheorem{proposition}[theorem]{Proposition}

\begin{document}

\title{Energy-Efficient Precoding for\\ Multiple-Antenna Terminals}


\author{Elena Veronica Belmega,~\IEEEmembership{Student~Member,~IEEE,}
and Samson Lasaulce,~\IEEEmembership{Member,~IEEE}

\thanks{Copyright (c) 2010 IEEE. Personal use of this material is permitted. However, permission to use this material for any other purposes must be obtained from the IEEE by sending a request to pubs-permissions@ieee.org.}

\thanks{E.~V. Belmega and S. Lasaulce are with LSS (joint lab of CNRS, Sup\'{e}lec, Paris 11), Sup\'{e}lec,
Plateau du Moulon, 91192 Gif-sur-Yvette, France,
\{belmega,lasaulce\}@lss.supelec.fr }}

\maketitle


\begin{abstract}
The problem of energy-efficient precoding is investigated when the terminals in the system are equipped with multiple antennas. Considering
static and fast-fading multiple-input multiple-output (MIMO) channels, the energy-efficiency is defined as the transmission rate to power ratio
and shown to be maximized at low transmit power. The most interesting case is the one of slow fading MIMO channels. For this type of channels,
the optimal precoding scheme is generally not trivial. Furthermore, using all the available transmit power is not always optimal in the sense of
energy-efficiency (which, in this case, corresponds to the communication-theoretic definition of the goodput-to-power (GPR) ratio). Finding the
optimal precoding matrices is shown to be a new open problem and is solved in several special cases: 1. when there is only one receive antenna;
2. in the low or high signal-to-noise ratio regime; 3. when uniform power allocation and the regime of large numbers of antennas are assumed. A
complete numerical analysis is provided to illustrate the derived results and stated conjectures. In particular, the impact of the number of
antennas on the energy-efficiency is assessed and shown to be significant.
\end{abstract}



\section{Introduction}
\label{sec:intro}

In many areas, like finance, economics or physics, a common way of
assessing the performance of a system is to consider the ratio of
what the system delivers to what it consumes. In communication
theory, transmit power and transmission rate are respectively two
common measures of the cost and benefit of a transmission.
Therefore, the ratio transmission rate (say in bit/s) to transmit
power (in J/s) appears to be a
 natural energy-efficiency measure of a communication system. An important question is then: what is the maximum amount of information (in
bits) that can be conveyed per Joule consumed? 
 As reported in \cite{verdu-it-1990}, one of the first papers addressing this issue is
\cite{pierce-tcom-1978} where the author determines the capacity per
unit cost for various versions of the photon counting channel. As
shown in \cite{verdu-it-1990}, the normalized\footnote{In
\cite{verdu-it-1990} the capacity per unit cost is in bit/s per
Joule and not in bit/J, which amounts to normalize by a quantity in
Hz.} capacity per unit cost for the well-known additive white
Gaussian channel model $Y = X + Z$ is maximized for Gaussian inputs
and is given by $\lim_{P \rightarrow 0} \frac{\log_2 \left(1+
\frac{P}{\sigma^2} \right)}{P} = \frac{1}{\sigma^2 \ln 2}$, where
$\mathbb{E}|X|^2 = P$ and $Z \sim \mathbb{C} \mc{N}(0, \sigma^2)$.
Here, the main message of communication theory to engineers is that
energy-efficiency is maximized by operating at low transmit power
and therefore at low transmission rates. However, this answer holds
for static and single input single output (SISO) channels and it is
legitimate to ask: what is the answer for multiple-input
multiple-output (MIMO) channels? In fact, as shown in this paper, the case of slow fading
MIMO channels is especially relevant to be considered. Roughly speaking, the main reason for this is that, in
contrast to static and fast fading channels, in slow fading
channels there are outage events which imply the existence of an optimum
tradeoff between the number of successfully transmitted bits or
blocks (called goodput in \cite{katz-it-2005} and
\cite{goodman-pcom-2000}) and power consumption. Intuitively, this
can be explained by saying that increasing transmit power too much
may result in a marginal increase in terms of quality or effective
transmission rate.

First, let us consider SISO slow fading or quasi-static channels. The most relevant works related to the problem under investigation essentially
fall into two classes corresponding to two different approaches. The first approach, which is the one adopted by Verd\'{u} in
\cite{verdu-it-1990} and has already been mentioned, is an information-theoretic approach aiming at evaluating the capacity per unit cost or the
minimum energy per bit (see e.g., \cite{elgamal-it-2006}, \cite{cai-tcom-2005}, \cite{yao-tw-2005}, \cite{jain-allerton-2009}). In
\cite{verdu-it-1990}, two different cases were investigated depending on whether the input alphabet contains or not a zero cost or free symbol.
In this paper, only the case where the input alphabet does not contain a zero-cost symbol will be discussed (i.e., the silence at the
transmitter side does not convey information). The second approach, introduced in \cite{shah-pimrc-1998} is more pragmatic than the previous
one. In \cite{shah-pimrc-1998} and subsequent works \cite{goodman-pcom-2000}, \cite{saraydar-tcom-2002}, the authors define the
energy-efficiency of a SISO communication as $u(p) = \frac{R f(\eta)}{p}$ where $R$ is the effective transmission data rate in bits, $\eta$ the
signal-to-noise-plus-interference ratio (SINR) and $f$ is a benefit function (e.g., the success probability of the transmission) which depends
on the chosen coding and modulation schemes. To the authors' knowledge, in all works using this approach (\cite{shah-pimrc-1998},
\cite{goodman-pcom-2000}, \cite{saraydar-tcom-2002}, \cite{meshkati-tc-2005}, \cite{buzzi-jsac-2008}, \cite{lasaulce-twc-2009}, etc.), the same
(pragmatic) choice is made for $f$: $f(x) = (1- e^{-\alpha x})^N$, where $\alpha$ is a constant and $N$ the block length in symbols.
Interestingly, the two mentioned approaches can be linked by making an appropriate choice for $f$. Indeed, if $f$ is chosen to be the
complementary of the outage probability, one obtains a counterpart of the capacity per unit cost for slow fading channels and gives an
information-theoretic interpretation to the initial definition of \cite{shah-pimrc-1998}. To our knowledge, the resulting performance metric has
not been considered so far in the literature. This specific metric, which we call goodput-to-power ratio (GPR), will be considered in this
paper. Moreover, we consider MIMO channels where the transmitter and receiver are informed of the channel distribution information (CDI) and
channel state information (CSI) respectively. To conclude the discussion on the relevant literature, we note that some authors addressed the
problem of energy-efficiency in MIMO communications but they did not consider the proposed energy-efficiency measure based on the outage
probability. In this respect, the most relevant works seem to be \cite{cui-jsac-2004}, \cite{verdu-it-2002} and \cite{buzzi-eusipco-2008}. In
\cite{cui-jsac-2004}, the authors adopt a pragmatic approach consisting in choosing a certain coding-modulation scheme in order to reach a given
target data rate while minimizing the consumed energy. In \cite{verdu-it-2002}, the authors study the tradeoff between the minimum
energy-per-bit versus spectral efficiency for several MIMO channel models in the wide-band regime assuming a zero cost symbol in the input
alphabet and unform power allocation over all the antennas. In \cite{buzzi-eusipco-2008}, the authors consider a similar pragmatic approach to
the one in
 \cite{goodman-pcom-2000}, \cite{saraydar-tcom-2002} and study a multi-user MIMO channel where the transmitters
 are constrained to using beamforming power allocation strategies.

This paper is structured as follows. In Sec. \ref{sec:sys}, assumptions on the signal model are provided. In Sec. \ref{sec:det_ff}, the proposed
energy-efficiency measure is defined for static and fast-fading MIMO channels. As the case of slow fading channels is non-trivial, it will be
discussed separately in Sec. \ref{sec:finite-MIMO}. In Sec. \ref{sec:finite-MIMO}, the problem of energy-efficient precoding is discussed for
general MIMO slow fading channels and solved for the multiple input single output (MISO) case, whereas in Sec. \ref{sec:asymptotic-MIMO}
asymptotic regimes (in terms of the number of antennas and SNR) are assumed. In Sec. \ref{sec:simus}, simulations illustrating the derived
results and stated conjectures are provided. Sec. \ref{sec:conclusion} provides concluding remarks and open issues.

%

\section{General System Model} \label{sec:sys}

We consider a point-to-point communication with multiple antenna
terminals. The signal at the receiver is modeled by:
\begin{equation}
\label{eq:system-model-mimo} \ul{y}(\tau)=\mb{H}(\tau) \ul{x}(\tau)
+ \ul{z}(\tau),
\end{equation}
where $\mb{H}$ is the $n_r \times n_t$ channel transfer matrix and $n_t$ (resp. $n_r$) the number of transmit (resp. receive) antennas. The
entries of $\mb{H}$ are i.i.d. zero-mean unit-variance complex Gaussian random variables. The vector $\ul{x}$ is the $n_t$-dimensional column
vector of transmitted symbols and $\ul{z}$ is an $n_r$-dimensional complex white
 Gaussian noise distributed as $\mathcal{N}(\ul{0}, \sigma^2
\mb{I})$. In this paper, the problem of allocating the transmit
power between the available transmit antennas is considered. We will
denote by $\mb{Q} = \mathbb{E}[\ul{x}\ul{x}^H]$ the input covariance
matrix
 (called the precoding matrix), which translates the chosen power
allocation (PA) policy. The corresponding total power constraint is
\begin{equation}
\label{eq:power-constraint} \mathrm{Tr}(\mb{Q}) \leq  \ol{P}.
\end{equation}
At last, the time index $\tau$ will be removed for the sake of
clarity. In fact, depending on the rate at which $\mb{H}$ varies with $\tau$, three
dominant classes of channel models can be distinguished:
\begin{enumerate}
    \item the class of
static channels;
    \item the class of fast fading channels;
    \item the class of slow fading channels.
\end{enumerate}
The matrix $\mb{H}$ is assumed to be perfectly known at the receiver (coherent communication assumption) whereas only the statistics of $\mb{H}$
are available at the transmitter. The first two classes of
 channels are considered in Sec. \ref{sec:det_ff} and the last one is treated
 in detail in Sec. \ref{sec:finite-MIMO}
and \ref{sec:asymptotic-MIMO}.

\section{Energy-efficient communications over
static and fast fading MIMO channels}
\label{sec:det_ff}
\subsection{Case of static channels}
\label{sec:sub-static-channels}


Here the frequency at which the channel matrix varies is strictly
 zero that is, $\mb{H}$ is a constant matrix. In this particular context, both the
transmitter and receiver are assumed to know this matrix. We are
exactly in the same framework as \cite{telatar-ett-1999}. Thus, for a given precoding scheme $\mb{Q}$, the transmitter can
send reliably to the receiver $\log_2\left|\mb{I}_{n_r}+ \rho
\mb{H}\mb{Q}\mb{H}^H\right|$ bits per channel use (bpcu) with $\rho
= \frac{1}{\sigma^2}$. Then, let us define the energy-efficiency of
this communication by:
\begin{equation}
\label{eq:def-ee-static-ch}
G_{\mathrm{static}}(\mb{Q})=\frac{\log_2\left|\mb{I}_{n_r}+ \rho \mb{H}\mb{Q}\mb{H}^H\right|}{\mathrm{Tr}(\mb{Q})}.
\end{equation}

The energy-efficiency $G_{\mathrm{static}}(\mb{Q})$ corresponds to
an achievable rate per unit cost for the MIMO channel as defined in
\cite{verdu-it-1990}. Assuming that the cost of the transmitted
symbol $\ul{x}$, denoted by $b[\ul{x}]$, is the consumed energy
$b[\ul{x}] = \|\ul{x}\|^2 =\mathrm{Tr}(\ul{x}\ul{x}^H)$, the capacity
per unit cost defined in \cite{verdu-it-1990} is:
$\widetilde{C}_{\mathrm{slow}} \triangleq \ds{\sup_{\ul{x},
\mathbb{E}[b[\ul{x}]] \leq \ol{P}} }
\frac{I(\ul{x};\ul{y})}{\mathbb{E}[b[\ul{x}]]}$. The supremum is
taken over the p.d.f. of $\ul{x}$ such that the average transmit
power is limited $\mathbb{E}[b[\ul{x}]] \leq \ol{P}$.

It is easy to check that:

\begin{equation}
\label{eq:it_det_mimo}
\begin{array}{lcl}
 \widetilde{C}_{\mathrm{slow}}  & = &  \ds{\sup_{\mb{Q}, \mathrm{Tr}(Q) \leq \ol{P}}} \frac{1}{\mathrm{Tr}(Q)} \ \ \ds{\sup_{\ul{x}, \mathbb{E}(\ul{x}\ul{x}^H) = \mb{Q}}
 }I(\ul{x};\ul{y}) \\
 & = & \ds{\sup_{\mb{Q}, \mathrm{Tr}(\mb{Q}) \leq \ol{P}} } G_{\mathrm{static}}(\mb{Q}).
\end{array}
\end{equation}

The second equality follows from \cite{telatar-ett-1999} where
Telatar proved that the mutual information for the MIMO static
channel is maximized using Gaussian random codes. In other words, finding
the optimal precoding matrix which maximizes the energy-efficiency
function corresponds to finding the capacity per unit cost of the
MIMO channel where the cost of a symbol is the necessary power
consumed to be transmitted. The question is then whether the
strategy ``transmit at low power'' (and therefore at a low
transmission rate) to maximize energy-efficiency, which is optimal
for SISO channels, also applies to MIMO channels. The answer is
given by the following proposition, which is proved in Appendix
\ref{appendix:A}.

\begin{proposition}[Static MIMO channels]
\label{proposition:static-channels} \emph{The energy-efficiency of a MIMO communication over a static channel, measured by
$G_{\mathrm{static}}$, is maximized when $\mb{Q} = \mb{0}$ and this maximum is}
\begin{equation}
G_{\mathrm{static}}^* = \frac{1}{\ln 2} \frac{\mathrm{Tr}(\mb{H}\mb{H}^H)}{n_t \sigma^2}.
\end{equation}
\end{proposition}

Therefore, we see that, for static MIMO channels, the energy-efficiency
defined in Eq. (\ref{eq:def-ee-static-ch}) is maximized by
transmitting at a very low power. This kind of scenario occurs for
example, when deploying sensors in the ocean to measure a
temperature field (which varies very slowly). In some applications
however, the rate obtained by using such a scheme can be not
sufficient. In this case, considering the benefit to cost ratio can
turn out to be irrelevant, meaning that other performance
metrics have to be considered (e.g., minimize the transmit power
under a rate constraint).

\subsection{Case of fast fading channels}
\label{sec:sub-fast-fading-ch}

In this section, the frequency with which the channel matrix varies is
the reciprocal of the symbol duration ($\ul{x}(\tau)$ being a
symbol). This means that it can be different for each channel use. Therefore,
the channel varies over a transmitted codeword (or packet) and, more
precisely, each codeword sees as many channel realizations as the
number of symbols per codeword. Because of the corresponding
self-averaging effect, the following transmission rate (also called
EMI for ergodic mutual information) can be achieved on each
transmitted codeword by using the precoding strategy $\mb{Q}$ :

\begin{equation}
R_{\mathrm{fast}}(\mb{Q}) = \mathbb{E}_{\mb{H}}\left[ \log_2\left|\mb{I}_{n_r}+ \rho \mb{H}\mb{Q}\mb{H}^H\right| \right].
\end{equation}

Interestingly, $R_{\mathrm{fast}}(\mb{Q})$ can be maximized w.r.t.
$\mb{Q}$ by knowing only the statistics of $\mb{H}$ that is,
$\mathbb{E}  \left[ \mb{H} \mb{H}^H \right]$, under the standard
assumption that the entries of $\mb{H}$ are complex Gaussian random
variables. In practice, this means that only the knowledge of the
path loss, power-delay profile, antenna correlation profile, etc is
required at the transmitter to maximize the transmission rate. At
the receiver however, the instantaneous knowledge of $\mb{H}$ is
required. In this framework, let us define energy-efficiency by:
\begin{equation}
\label{eq:def-ee-fast-fading-ch} G_{\mathrm{fast}}(\mb{Q})=\frac{\mathbb{E}_{\mb{H}} \left[\log_2\left|\mb{I}_{n_r}+\rho
\mb{H}\mb{Q}\mb{H}^H\right|\right]}{\mathrm{Tr}(\mb{Q})}.
\end{equation}
By defining $\ul{g}_i$ as the $i$-th column of the matrix
$\sqrt{\rho} \mb{H}\mb{U}$, $i \in \{1, \hdots, n_t \}$, $\mb{U}$ and $\{p_i\}_{i=1}^{n_t}$ an
eigenvector matrix and the corresponding eigenvalues of $\mb{Q}$ respectively, and also by rewriting $G_{\mathrm{fast}}(\mb{Q})$
as
\begin{equation}
\ds{G_{\mathrm{fast}}(\mb{Q}) =\mathbb{E}_{\mb{H}}\left[
\frac{\log_2\left|\mb{I}_{n_r}+\ds{\sum_{i=1}^{n_t} p_i \ul{g}_i
\ul{g}_i^H} \right|}{\ds{\sum_{i=1}^{n_t}p_i}}\right]},
\end{equation}
it is possible to apply the proof of Prop.
\ref{proposition:static-channels} for each realization of the
channel matrix. This leads to the following result.

\begin{proposition}[Fast fading MIMO channels]
\label{proposition:fast-fading-channels} T\emph{he energy-efficiency of a MIMO communication over a fast fading channel, measured by
$G_{\mathrm{fast}}$, is maximized when $\mb{Q} = \mb{0}$ and this maximum is}
\begin{equation}
G_{\mathrm{fast}}^* = \frac{1}{\ln 2} \frac{\mathrm{Tr}(\mathbb{E} \left[ \mb{H}\mb{H}^H \right])}{n_t \sigma^2}.
\end{equation}
\end{proposition}
We see that, for fast fading MIMO channels, maximizing
energy-efficiency also amounts to transmitting at low power.
Interestingly, in slow fading MIMO channels, where outage events are
unavoidable, we have found that the answer can be different. This is
precisely what is shown in the remaining of this paper.


\section{Slow fading MIMO channels: from the general case to
special cases} \label{sec:finite-MIMO}

\subsection{General MIMO channels}

In this section and the remaining of this paper, the frequency with
which the channel matrix varies is the reciprocal of the
block/codeword/frame/packet/time-slot duration that is, the channel
remains constant over a codeword and varies from block to block. As
a consequence, when the channel matrix remains constant over a
certain block duration much smaller than the channel coherence time,
the averaging effect we have mentioned for fast fading MIMO channels
does not occur here. Therefore, one has to communicate at rates smaller than
the ergodic capacity (maximum of the EMI). The maximum EMI is
therefore a rate upper bound for slow fading MIMO channels and only
a fraction of it can be achieved (see \cite{zheng-allerton-2001} for
more information about the famous diversity-multiplexing tradeoff).
In fact, since the mutual information is a random variable,
varying from block to block, it is not possible (in general) to
guarantee at $100 \ \%$ that it is above a certain threshold. A
suited performance metric to study slow-fading channels
\cite{ozarow-vt-1994}
 is the probability of
an outage for a given transmission rate target $R$. This metric allows one to
quantify the probability that the rate target $R$ is not reached by
using a good channel coding scheme and is defined as follows:
\begin{equation}
\label{eq:def-outage-proba} \mathrm{P}_{\mathrm{out}}(\mb{Q},R)=\mathrm{Pr}\left[\log_2 \left|\mb{I}_{n_r}+ \rho\mb{H}\mb{Q}\mb{H}^H\right|<R\right].
\end{equation}
In terms of information assumptions, here again, it can be checked
that only the second-order statistics of $\mb{H}$ are required to
optimize the precoding matrix $\mb{Q}$ (and therefore the power
allocation policy over its eigenvalues). In this framework, we
propose to define the energy-efficiency as follows:
\begin{equation}
\label{eq:gen_payoff} \Gamma(\mb{Q},R)=  \frac{R [ 1-
\mathrm{P}_{\mathrm{out}}(\mb{Q},R)]}{\mathrm{Tr}(\mb{Q})}.
\end{equation}

In other words, the energy-efficiency or goodput-to-power ratio is defined as the ratio
between the expected throughput (see
\cite{katz-it-2005},\cite{shamai-ciss-2000} for details) and the
average consumed transmit power. The expected throughput can be seen
as the average system throughput over many transmissions. In
contrast with static and fast fading channels, energy-efficiency is
not necessarily maximized at low transmit powers. This is what the
following proposition indicates.
\begin{proposition}[Slow fading MIMO channels] \label{proposition:slow_gen}
\emph{The goodput-to-power ratio $\Gamma(\mb{Q},R)$ is maximized, in general, for $\mb{Q} \neq \mb{0}$.}
\end{proposition}
The proof of this result is given in Appendix \ref{appendix:B}. Now, a natural issue to be considered is the determination of
  the matrix (or matrices) maximizing the goodput-to-power ratio (GPR) in slow
   fading MIMO channels. It turns out that the corresponding
    optimization problem is not trivial.
Indeed, even the outage probability minimization problem w.r.t.
$\mb{Q}$ (which is a priori simpler) is still an open problem
\cite{telatar-ett-1999}, \cite{katz-wc-2007}, \cite{jorswieck-ett-2007}.
This is why we only provide here a conjecture on the solution
maximizing the GPR.

\begin{conjecture}[Optimal precoding matrices]
\label{conjecture:GPR_MIMO}
 \emph{There exists a power
threshold $\ol{P}_0$ such that:}
\begin{itemize}
    \item \emph{if $\ol{P} \leq \ol{P}_0$ then $\mb{Q}^* \in \ds{ \arg
    \min_{\mb{Q}}} P_{\mathrm{out}}(\mb{Q},R)$ $ \ \Rightarrow \ $
    $\mb{Q}^* \in \ds{ \arg \max_{\mb{Q}} } \Gamma(\mb{Q},R)$;}

    \item \emph{if $\ol{P} > \ol{P}_0$ then $\Gamma(\mb{Q},R)$
    has a unique maximum in $\mb{Q}^* = \frac{p^*}{n_t} \mb{I}_{n_t}$ where
    $p^* \leq \ol{P}$.}
\end{itemize}
\end{conjecture}

This conjecture has been validated for all the special cases solved in this paper. One of the main messages of this conjecture is that, if the
available transmit power is less than a threshold, maximizing the GPR is equivalent to minimizing the outage probability. If it is above the
threshold, uniform power allocation is optimal and using all the available power is generally suboptimal in terms of energy-efficiency.
Concerning the optimization problem associated with (\ref{eq:gen_payoff}) several
 comments are in order. First, there is no loss of optimality
by restricting the search for optimal precoding matrices to diagonal
matrices: for any eigenvalue decomposition $\mb{Q} = \mb{U} \mb{D}
\mb{U}^H$ with $\mb{U}$ unitary and $\mb{D} =
\mathrm{\textbf{Diag}}(\ul{p})$ with $\ul{p}= (p_1, \hdots,
p_{n_t})$, both the outage and trace are invariant w.r.t. the choice
of $\mb{U}$ and the energy-efficiency can be written as:

\begin{equation}
\label{eq:diag_gen_payoff} \Gamma(\mb{D},R)=  \frac{R [ 1- \mathrm{P}_{\mathrm{out}}(\mb{D},R)]}{\ds{\sum_{i=1}^{n_t}} p_i}.
\end{equation}

Second, the GPR is generally not concave w.r.t. $\mb{D}$. In Sec.
\ref{sec:sub_MISO}, which is dedicated to MISO systems, a
counter-example where it is not quasi-concave (and thus not concave)
is provided.

\emph{Uniform Power Allocation policy}

An interesting special case is the one of uniform power allocation
(UPA): $\mb{D} = \frac{p}{n_t} \mb{I}_{n_t}$ where $p \in [0, \ol{P}]$ and
$\Gamma_{\mathrm{UPA}}(p,R)\triangleq
\Gamma\left(\frac{p}{n_t}\mb{I}_{n_t},R\right)$.

One of the reasons for
studying this case is that the famous conjecture of Telatar given in
\cite{telatar-ett-1999}. This conjecture states that, depending on the channel
parameters and target rate (i.e., $\sigma^2$, $R$), the power allocation (PA) policy
minimizing the outage probability is to spread all the available
power uniformly over a subset of $\ell^* \in \{1, \hdots, n_t\}$
antennas. If this can be proved, then it is straightforward to show
that the covariance matrix $\mb{D}^*$ that maximizes the proposed
energy-efficiency function is $\frac{p^*}{\ell^*}
\mathrm{\textbf{Diag}}(\ul{e}_{\ell^*})$, where $\ul{e}_{\ell^*} \in \mc{S}_{\ell^*}$\footnote{We denote by
$\mc{S}_{\ell} = \left\{\ul{v} \in \{0,1\}^{n_t} | \sum_{i=1}^{n_t} v_i = \ell \right\}$ the set of $n_t$ dimensional vectors containing $\ell$ ones and $n_t - \ell$ zeros, for all $\ell \in \{1,\hdots,n_t\}$.}. Thus, $\mb{D}^*$ has the same structure as
the covariance matrix minimizing the outage probability except that
using all the available power is not necessarily optimal, $p^* \in
[0, \ol{P}]$. In conclusion, solving Conjecture
\ref{conjecture:GPR_MIMO} reduces to solving Telatar's conjecture
and also the UPA case.

The main difficulty in studying the outage
probability or/and the energy-efficiency function is the fact that
the probability distribution function of the mutual information is
generally intractable. In the literature, the outage probability is
often studied by assuming a UPA policy over all the antennas and
also using the Gaussian approximation of the p.d.f. of the mutual information.
This approximation is valid in the asymptotic regime of large number
of antennas. However, simulations show that
 it also quite accurate for reasonable small MIMO systems
 \cite{wang-it-2004}, \cite{moustakas-it-2003}.

 Under the UPA policy
 assumption, the GPR $\Gamma_{\mathrm{UPA}}(p,R)$ is conjectured
 to be quasi-concave w.r.t. $p$. Quasi-concavity is not only
 useful to study the maximum of the GPR but is also an attractive
  property in some scenarios such as the distributed
multiuser channels. For example, by considering MIMO multiple
access channels with single-user decoding at the receiver,
the corresponding
distributed power allocation game where the transmitters' utility
 functions are their GPR is guaranteed to have a pure Nash equilibrium after
Debreu-Fan-Glicksberg theorem \cite{fundenberg-book-1991}.

Before
stating the conjecture describing the behavior of the
energy-efficiency function when the UPA policy is assumed, we study
the limits when $p \rightarrow 0$ and $p \rightarrow + \infty.$
First, let us prove that $\ds{\lim_{p \rightarrow 0}
\Gamma_{\mathrm{UPA}}(p,R) = 0}$. Observe that $\ds{\lim_{p
\rightarrow 0} P_{\mathrm{out}}\left(\frac{p}{n_t}\mb{I}_{n_t},R\right) =
1}$ and thus the limit is not trivial to prove. The result can be
proven by considering the equivalent $1+\frac{\rho p}{n_t}
\mathrm{Tr}(\mb{H}\mb{H}^H)$ of the determinant
$\left|\mb{I}_{n_r}+\frac{\rho p}{n_t} \mb{H}\mb{H}^H \right|$ when
$\sigma \rightarrow + \infty$. As the entries of the matrix $\mb{H}$
are i.i.d. complex Gaussian random variables, the quantity
$\mathrm{Tr}(\mb{H}\mb{H}^H)= \ds{\sum_{i=1}^{n_t} \sum_{j=1}^{n_r}}
|h_{ij}|^2$ is a $2 n_r n_t$ Chi-square distributed random variable.
Thus $\Gamma_{\mathrm{UPA}}(p,R)$ can be approximated by:
$\widehat{\Gamma}_{\mathrm{UPA}}(p,R)=R \exp
\left(-\frac{d}{p}\right) \ds{ \sum_{k=0}^{n_r n_t -1} }
\frac{d^k}{k!} \frac{1}{p^{k+1}}$ with $d = n_t (2^R-1)\sigma^2$. It
is easy to see that this approximate tends to zero when $p
\rightarrow 0$. Second, note that the limit $\ds{\lim_{p \rightarrow
+ \infty} \Gamma_{\mathrm{UPA}}(p, R) = 0}$. This is easier to check
since $\ds{\lim_{p \rightarrow +\infty}
P_{\mathrm{out}}\left(\frac{p}{n_t}\mb{I},R\right) = 0}$.

\begin{conjecture}[UPA and quasi-concavity of the GPR]
\label{conjecture:MIMO_UPA} \emph{Assume that
$\mb{D}=\frac{p}{n_t}\mb{I}_{n_t}$. Then $\Gamma_{\mathrm{UPA}}(p,R)$ is
quasi-concave w.r.t. $p \in \left[0, \ol{P}\right]$.}
\end{conjecture}
Table \ref{table_results} distinguishes between what has been proven
in this paper and the conjectures which remain to be proven.

\begin{table}
\label{table_results}
\begin{center}
\begin{tabular}{|c||c|c|c|}
  \hline
    & Is $\mb{D}^*$ known? & Is
    $\Gamma^{\mathrm{UPA}}(p)$ quasi-concave? & Is $p^*$ known? \\
\hline \hline
  General MIMO  & Conjecture    & Conjecture & Conjecture \\
\hline
  MISO          & Yes           & Yes             & Yes \\
\hline
  $1 \times 1$  & Yes           & Yes             & Yes \\
\hline
  Large MIMO    & Conjecture    & Yes             & Yes \\
\hline
  Low SNR       & Yes    & Yes             & Yes \\
\hline
  High SNR      & Yes   & Yes             & Conjecture \\
  \hline
\end{tabular}
\caption{Summary of proved results and open problems}
\end{center}
\end{table}

\subsection{MISO channels}
\label{sec:sub_MISO}

In this section, the receiver is assumed to use a single antenna
that is, $n_r  =1$, while the transmitter can have an arbitrary
number of antennas, $n_t \geq 1$. The channel transfer matrix becomes
a row vector $\ul{h} = (h_1,...,h_{n_t})$. Without loss of
optimality, the precoding matrix is assumed to be diagonal and is denoted by $\mb{D}
 = \mathrm{\textbf{Diag}}(\ul{p})$ with $\ul{p}^T= (p_1,...,p_{n_t})$. Throughout
  this
 section, the rate target $R$ and noise level $\sigma^2$ are fixed and
 the auxiliary quantity $c$ is defined by: $c = \sigma^2 (2^R -1)$. By
exploiting the existing results on the outage probability
minimization problem for MISO channels \cite{jorswieck-ett-2007},
the following proposition can be proved (Appendix \ref{appendix:C}).

\begin{proposition}[Optimum precoding matrices for MISO channels]
\label{proposition:MISO} \emph{For all $\ell \in \{1,...,n_t-1\}$, let $c_{\ell}$ be the unique solution of the equation (in $x$)
$\mathrm{Pr}\left[\frac{1}{\ell+1} \ds{\sum_{i=1}^{\ell+1}} |X_i|^2 \leq x  \right] - \mathrm{Pr}\left[\frac{1}{\ell} \ds{\sum_{i=1}^{\ell}}
|X_i|^2 \leq x \right] = 0$ where $X_i$ are i.i.d. zero-mean Gaussian random variables with unit variance. By convention $c_0 = + \infty$,
$c_{n_t} = 0$. Let $\nu_{n_t}$ be the unique solution of the equation (in $y$) $\frac{y^{n_t}}{(n_t-1)!}  - \ds{\sum_{i=0}^{n_t-1}}
\frac{y^i}{i!} =0$. Then the optimum precoding matrices have the following form:}
\begin{equation}
\mb{D}^* = \left|
\begin{array}{cl}
 \frac{\ol{P}}{\ell} \mathrm{\textbf{Diag}}(\ul{e}_{\ell}) & \ \mathrm{if} \ \ol{P} \in \left[\frac{c}{c_{\ell-1}},
  \frac{c}{c_{\ell}}
   \right) \\
   \min\left\{\frac{\sigma^2 (2^R-1) }{\nu_{n_t}},
   \frac{\ol{P}}{n_t}  \right\} \mb{I} & \ \mathrm{if} \  \ol{P} \geq
 \frac{c}{c_{n_t-1}}
\end{array}
\right.
\end{equation}
\emph{where $c = \sigma^2 (2^R-1) $ and $\ul{e}_{\ell} \in \mc{S}_{\ell}$.}
\end{proposition}
Similarly to the optimal precoding scheme for
 the outage probability minimization, the solution maximizing the
 GPR consists in allocating the available transmit power uniformly between
only a subset $\ell \leq n_t$ antennas. As i.i.d entries are assumed
for $\mb{H}$, the choice
 of these antennas does not matter. What matters is the number of
 antennas selected (denoted by $\ell$), which depends on the available transmit power
 $\ol{P}$: the higher the transmit power, the higher the number of
 used antennas. The difference between the outage probability
 minimization and GPR maximization problems appears when the
 transmit power is greater than the threshold $\frac{c}{c_{n_t-1}}$.
 In this regime, saturating the power constraint is suboptimal for the GPR optimization. The
 corresponding sub-optimality becomes more and more severe as the
 noise level is low; simulations (Sec. \ref{sec:simus}) will help us to quantify
 this gap.

 Unless otherwise specified, we will assume from now on that
\textbf{UPA} is used at the transmitter. This assumption is, in
particular, useful
to study the regime where the available transmit
power is sufficiently high (as conjectured in
Proposition \ref{proposition:slow_gen}). 
  Under this assumption, our goal is to prove
  that the GPR is quasi-concave w.r.t. $p \in [0, \ol{P}]$ with $\mb{D} =
\frac{p}{n_t} \mb{I}_{n_t}$ and determine the (unique) solution $p^*$
which maximizes the GPR. Note that the quasi-concavity property
w.r.t. $\ul{p}$ is not always available for MISO systems (and thus
is not always available for general MIMO channels). In Appendix
\ref{appendix:D}, a counter-example proving that in the case where
$n_r=1$ and $n_t=2$ (two input single output channel, TISO) the
energy-efficiency $\Gamma^{\mathrm{TISO}}\left(\textbf{Diag}(\ul{p}), R\right)$ is not
quasi-concave w.r.t. $\ul{p}=(p_1,p_2)$ is provided.

\begin{proposition}[UPA and quasi-concavity (MISO channels)]
\label{proposition:MISO_UPA}
Assume the UPA, $\mb{Q}=\frac{p}{n_t}\mb{I}_{n_t}$, then $\Gamma(p,R)$ is quasi-concave w.r.t. $p \in \left[0, \ol{P}\right]$ and has a unique maximum point in
$p^* = \min \left\{\frac{(2^R-1)n_t \sigma^2}{\nu_{n_t}}, \ol{P} \right\}$ where $\nu_{n_t}$ is the solution (w.r.t. $y$) of:
\begin{equation}
\label{eq:y}
\frac{y^{n_t}}{(n_t-1)!} - \ds{\sum_{i=0}^{n_t-1}} \frac{y^i}{i!} = 0.
\end{equation}
\end{proposition}
\begin{proof}
Since the entries of $\ul{h}$ are complex Gaussian random variables, the sum $\ds{\sum_{k=1}^{n_t}} |h_k|^2$ is a $2 n_t-$
Chi-square distributed random variable, which implies that:
\begin{equation}
\begin{array}{lcl}
\Gamma^{\mr{MISO}}(p,R) & = & \ds{\frac{R \left\{1 -
\mathrm{Pr}[\log_2\left(1+\frac{\rho p}{n_t} \ul{h}^H\ul{h}\right) <
R] \right\} }{p}}\\
& = & \ds{\frac{R \left\{1 -  \mathrm{Pr} \left[ \ds{\sum_{i=1}^{n_t}} |h_i|^2 < \frac{d}{p}\right]\right\}}{p }} \\
& = & \ds{R \times \mathrm{e}^{-\frac{d}{p}} \sum_{i=0}^{n_t-1}
\frac{d^i}{p^{i+1}} \frac{1}{i!}},
\end{array}
\end{equation}
with $d =c n_t = (2^R-1)n_t \sigma^2$.
The second order derivative of the goodput $R \left[\mathrm{e}^{-\frac{d }{p}} \ds{\sum_{i=0}^{n_t-1} } \left(\frac{d}{p}\right)^i
\frac{1}{i!}\right]$ w.r.t. $p$ is \newline $R \left[\frac{d^{n_t}}{p^{n_t+3}} \frac{1}{n_t!} \mathrm{e}^{-d/p} (d- (n_t+1)p)\right]$. Clearly, the
goodput is a sigmoidal function and has a unique inflection point in $p_0 = \frac{d}{n_t+1}$. Therefore, the function $\Gamma^{\mr{MISO}}(p,R)$
is quasi-concave \cite{rodriguez-globecom-2003} and has a unique maximum in $p^* = \min \left\{\frac{d}{\nu_{n_t}}, \ol{P} \right\}$ where $\nu_{n_t}$ is the root of the first order derivative of
$\Gamma^{\mr{MISO}}(p,R)$ that is, the solution of (\ref{eq:y}).
\end{proof}
The \textbf{SIMO} case ($n_t=1$, $n_r \geq 2$) follows directly
since $|\mb{I}+ \rho p \ul{h}\ul{h}^H| = 1+ \rho p \ul{h}^H\ul{h}$.

To conclude this section, we consider the most simple case of MISO channels namely the SISO case ($n_t=1$, $n_r=1$). We have readily that:
\begin{equation}
\label{eq:siso-GPR}
    \Gamma^{\mr{SISO}}(p,R) =  \frac{e^{-\frac{c}{p}}}{p}.
\end{equation}
To the authors' knowledge, in all the works using the
energy-efficiency definition of \cite{goodman-pcom-2000} for SISO
channels, the only choice of energy-efficiency function made is
based on the empirical approximation of the block error rate which
is $\frac{(1- e^{-x})^M}{x}$, $M$ being the block length and $x$ the
operating SINR. Interestingly, the function given by
(\ref{eq:siso-GPR}) exhibits another possible choice. It can be
checked that the function $e^{-\frac{c}{p}}$ is sigmoidal and
therefore $\Gamma^{\mr{SISO}}$ is quasi-concave w.r.t. $p$
\cite{rodriguez-globecom-2003}. The first order derivative of
$\Gamma^{\mr{SISO}}$ is
\begin{equation}
\frac{\partial \Gamma^{\mr{SISO}}}{\partial p} = R \frac{(c-p)
e^{-\frac{c}{p}}}{p^3}.
\end{equation}
The GPR is therefore maximized in a unique point which $p^* = c = \sigma^2 (2^R-1)$. To make the bridge between this solution and the one
derived in \cite{goodman-pcom-2000} for the power control problem over multiple access channels, the optimal power level can be rewritten as:

\begin{equation}
\label{eq:pstar_SISO} p^* = \min
\left\{\frac{\sigma^2}{\mathbb{E}|h|^2} (2^R-1), \ol{P} \right\}
\end{equation}

where $\mathbb{E}|h|^2=1$ in our case. In \cite{goodman-pcom-2000}, instantaneous CSI knowledge at the transmitters is assumed while here only
the statistics are assumed to be known at the transmitter. Therefore, the power control interpretation of (\ref{eq:pstar_SISO}) in a wireless
scenario is that the power is adapted to the path loss (slow power control) and not to fast fading (fast power control).

\section{Slow fading MIMO channels in asymptotic regimes }
\label{sec:asymptotic-MIMO}

In this section, we first consider the GPR for the case where the
size of the MIMO system is finite assuming the low/high SNR
operating regime. Then, we consider the UPA policy and prove that
Conjecture \ref{conjecture:MIMO_UPA} claiming that
$\Gamma_{\mr{UPA}}(p,R)$ is quasi-concave w.r.t. $p$ (which has been
proven for MISO, SIMO, and SISO channels) is also valid in the
asymptotic regimes where either at least one dimension of the system
($n_t$, $n_r$) is large but the SNR is finite. Here again, the
theory of large random matrices is successfully applied since it
allows one to prove some results which are not available yet in the
finite case (see e.g., \cite{tse-it-2003}, \cite{dumont-arxiv-2007}
for other successful examples).

\subsection{Extreme SNR regimes}

Here, all the channel parameters ($n_t$, $n_r$, and $\ol{P}$ in
particular) are fixed. The low (resp. high) SNR regime is defined by
$\sigma^2 \rightarrow + \infty$ (resp. $\sigma^2 \rightarrow  0$).
In both cases, we will consider the GPR and the optimal power
allocation problem.

\subsubsection{Low SNR regime}

Let us consider the general power allocation problem where
$\mb{D}=\mathrm{\textbf{Diag}}(\ul{p})$ with
$\ul{p}=(p_1,\hdots,p_{n_t})$. In \cite{jorswieck-ett-2007}, the
authors extended the results obtained in the low and high SNR
regimes for the MISO channel to the MIMO case. In the low SNR
regime, the authors of \cite{jorswieck-ett-2007} proved that the
outage probability
$P_{\mathrm{out}}(\mathrm{\textbf{Diag}}(\ul{p}),R)$ is a
Schur-concave (see \cite{marshall-book-1979} for details) function
w.r.t. $\ul{p}$. This
implies directly that beamforming power allocation policy maximizes
the outage probability. These results can be used (see Appendix
\ref{appendix:E}) to prove the following proposition:

\begin{proposition}[Low SNR regime]
\label{proposition:low-snr} \emph{When $\sigma^2 \rightarrow +
\infty$, the energy-efficiency function}
$\Gamma(\mathrm{\textbf{Diag}}(\ul{p}), R)$ \emph{is Schur-concave
w.r.t. $\ul{p}$ and maximized by a beamforming power allocation
policy} $\mb{D}^* = \ol{P}  \mathrm{\textbf{Diag}} (\ul{e}_1)$.
\end{proposition}

\subsubsection{High SNR regime}

Now, let us consider the high SNR regime. It turns out that the UPA
policy maximizes the energy-efficiency function. In this case also,
the proof of the following proposition is based on the results in
\cite{jorswieck-ett-2007} (see Appendix \ref{appendix:E}).

\begin{proposition}[High SNR regime]
\label{proposition:low-snr}\emph{ When $\sigma^2 \rightarrow 0$, the
energy-efficiency function} $\Gamma(\mathrm{\textbf{Diag}}(\ul{p}),
R)$ \emph{is
 Schur-convex w.r.t. $\ul{p}$ and maximized by an
 uniform power allocation policy $\mb{D}^*=\frac{p^*}{n_t}\mb{I}_{n_t}$ with $p^* \in (0, \ol{P}]$.
 Furthermore, the limit when $p \rightarrow 0$ such that $\frac{p}{\sigma^2} \rightarrow \xi$ is
 $\Gamma\left( \frac{p}{n_t}\mb{I}_{n_t}, R \right) \rightarrow + \infty$ which implies that $p^* \rightarrow 0$.}
\end{proposition}

In other words, in the high SNR regime, the optimal structure of the
covariance matrix is obtained by uniformly spreading the power over
all the antennas, $\mb{D}^*=\frac{p^*}{n_t}\mb{I}_{n_t}$ the same
structure which minimizes the outage probability in this case.
Nevertheless, in contrast to the outage probability optimization
problem, in order to be energy-efficient it is not optimal to use
all the available power $\ol{P}$ but to transmit with zero power.

\subsection{Large MIMO channels}
\label{subsec:asym_MIMO}

The results we have obtained can be summarized in the following
proposition.

\begin{proposition}[Quasi-concavity for large MIMO systems]
\label{proposition:asymptotics} \emph{If the system operates in one
of the following asymptotic regimes:}
\begin{description}
    \item{(a)} $n_t < + \infty$ and $n_r \rightarrow +\infty$;
    \item{(b)} $n_t \rightarrow +\infty$ and $n_r < + \infty$;
    \item{(c)} $n_t \rightarrow +\infty$, $n_r \rightarrow +\infty$ with
    $\ds{\lim_{n_i \rightarrow +\infty, i \in \{t,r\}} \frac{n_r}{n_t} = \beta < + \infty}$,
\end{description}
\emph{then $\Gamma_{\mr{UPA}}(p, R)$ is quasi-concave w.r.t. $p \in [0, \ol{P}]$.}
\end{proposition}


\begin{proof}
Here we prove each of the three statements made above and provide
comments on each of them at the same time.

\emph{Regime (a): $n_t < + \infty$ and $n_r \rightarrow \infty$.} The idea of the proof is to consider a large system equivalent of the function
$\Gamma_{\mr{UPA}}(p,R)$. This equivalent is denoted by $\widehat{\Gamma}^{\mr{a}}_{\mr{UPA}}(p,R)$ and is based on the Gaussian approximation
of the mutual information $\log_2\left|\mb{I}+\frac{\rho p}{n_t} \mb{H}\mb{H}^H\right|$ (see e.g., \cite{hochwald-it-2004}). The goal is to
prove that the numerator of $\widehat{\Gamma}^{\mr{a}}_{\mr{UPA}}(p,R)$ is a sigmoidal function w.r.t. $p$ which implies that
$\widehat{\Gamma}^{\mr{a}}_{\mr{UPA}}(p,R)$ is a quasi-concave function \cite{rodriguez-globecom-2003}. In the considered asymptotic regime, we
know from \cite{hochwald-it-2004} that:
\begin{equation}
\log_2\left|\mb{I}+\frac{\rho p}{n_t} \mb{H}\mb{H}^H\right|
\rightarrow \mc{N}\left(n_t \log_2\left(1+\frac{n_r}{n_t}\rho
 p \right), \frac{n_t}{n_r} \log_2(e)\right).
\end{equation}
A large system equivalent of the numerator of $\Gamma_{\mr{UPA}}(p,R)$, which is denoted by $\widehat{N}_a(p,R)$, follows:
\begin{equation}
\label{eq:f-large-mimo-a} \widehat{N}_a(p,R) =R Q\left(\frac{R-n_t
\log_2\left(1+\frac{n_r}{n_t}\rho
 p \right)}{\sqrt{\frac{n_t}{n_r}\log_2(e)}}\right)
\end{equation}
where $Q(x)= \frac{1}{\sqrt{2 \pi }}\int_{x}^{+\infty} \exp \left(-\frac{t^2}{2}\right) \mr{d}t$. Denote the argument of $Q$ in
(\ref{eq:f-large-mimo-a}) by $\alpha_a$. The second order derivative of $\widehat{N}_a(p,R)$ w.r.t. $p$
\begin{equation}
\label{eq:cp_asym_secd} \frac{\partial^2 \widehat{N}_a(p,R)}{\partial p^2} = \frac{1}{\sqrt{2
\pi}}\left[\alpha_a(p)(\alpha_a'(p))^2-\alpha_a''(p)\right]\exp\left(-\frac{\alpha_a(p)^2}{2}\right).
\end{equation}
Therefore $\widehat{N}_a(p,R)$ has a unique inflection point
\begin{equation}
\tilde{p}_a=\frac{n_t}{n_r\rho}\left\{2^{\left[\frac{1}{n_t}\left(R-\frac{1}{n_t}\left(\frac{n_t
\log_2(e)}{n_r}\right)^{3/2}\right)\right]} -1 \right\}.
\end{equation}
Clearly, for each equivalent of $\Gamma_{\mr{UPA}}(p,R)$, the numerator has a unique inflection point and is sigmoidal, which concludes the
proof. In fact, in the considered asymptotic regime we have a stronger result since $\ds{\lim_{n_r \rightarrow + \infty}}\tilde{p}_a = 0$, which
implies that $\widehat{N}_a(p,R)$ is concave and therefore $\widehat{\Gamma}^{\mr{a}}_{\mr{UPA}}(p,R)$ is maximized in $p_a^* =0$ as in the case
of static MIMO channels. This translates the well-known channel hardening effect \cite{hochwald-it-2004}. However, in contrast to the static case,
the energy-efficiency becomes infinite here since $\Gamma_{\mr{UPA}}(p,R) \rightarrow \frac{1}{p} $ with $p_a^* \rightarrow 0$.

\emph{Regime (b): $n_t \rightarrow + \infty$ and $n_r < + \infty$.} To prove the corresponding result the same reasoning as in (a) is applied.
From \cite{hochwald-it-2004} we know that:
\begin{equation}
\log_2\left|\mb{I}+\frac{\rho p}{n_t} \mb{H}\mb{H}^H\right| \rightarrow \mc{N}\left(n_r \log_2(1+\rho p) , \left(\sqrt{\frac{n_r}{n_t}}
\log_2(e)\frac{\rho p}{1+\rho p}\right)^2\right).
\end{equation}
A large system equivalent of the numerator of $\Gamma_{\mr{UPA}}(p,R)$ is $\widehat{N}_b(p,R) = R Q\left(\alpha_b(p) \right)$ with
\begin{equation}
\alpha_b(p)= \sqrt{\frac{n_t}{n_r}}\log_2(e)\frac{1+\rho p}{\rho p}[R-n_r\log_2(1+\rho p)].
\end{equation}
The numerator function $\widehat{N}_b(p,R)$ can be checked to have a
unique inflection point given by:
\begin{equation}
\tilde{p}_b = \sigma^2 \left(2^{\frac{R}{n_r}} -1\right)
\end{equation}
and is sigmoidal, which concludes the proof. We see that the inflection point does not vanish this time (with $n_t$ here) and
therefore the function $\widehat{N}_b(p,R)$ is quasi-concave but not
concave in general. From \cite{rodriguez-globecom-2003}, we know
that the optimal solution $p^*_b$ represents the point where the
tangent that passes through the origin intersects the S-shaped
function $R Q\left(\alpha_b(p) \right)$. As $n_t$ grows large, the
function $Q\left(\alpha_b(p) \right)$ becomes a Heavyside step
function since $\forall p \leq \tilde{p}_b$, $\lim_{n_t\rightarrow +
\infty} Q\left(\alpha_b(p) \right) = 0$ and $\forall p \geq
\tilde{p}_b$, $\lim_{n_t\rightarrow + \infty} Q\left(\alpha_b(p)
\right) = 1$. This means that the optimal power $p^*_b$ that
maximizes the energy-efficiency approaches $\tilde{p}_b$ as $n_t$
grows large, $p^*\begin{scriptsize}\begin{footnotesize}\begin{small}\end{small}\end{footnotesize}\end{scriptsize}_b \rightarrow \sigma^2 \left(2^{\frac{R}{n_r}}
-1\right)$. The optimal energy-efficiency tends to
$\frac{\widehat{N}_b(p^*_b,R)}{p^*_b} \rightarrow \frac{1}{2
\sigma^2 \left(2^{\frac{R}{n_r}} -1\right)}$ when $n_t \rightarrow +
\infty$.

\emph{Regime (c): $n_t \rightarrow +\infty$, $n_r \rightarrow \infty$.} Here we always apply the same reasoning but exploit the results derived
in \cite{debbah-it-2005}. From \cite{debbah-it-2005}, we have that:
\begin{equation}
\log_2\left|\mb{I}+\frac{\rho p}{n_t} \mb{H}\mb{H}^H\right|
\rightarrow \mc{N}\left(n_t \mu_I, \sigma_I^2\right)
\end{equation}
where $ \mu_I=\beta \log_2(1+\rho p(1-\gamma)) - \gamma +
\log_2(1+\rho p (\beta-\gamma)) $,
$\sigma_I^2=-\log_2\left(1-\frac{\gamma^2}{\beta}\right)$, \newline
$\gamma= \frac{1}{2}\left(1+\beta+\frac{1}{\rho p}-
\sqrt{(1+\beta+\frac{1}{\rho p})^2-4\beta}\right)$. It can be
checked that $(\alpha_c'(p))^2 \alpha_c(p)-\alpha_c''(p)=0$ has a
unique solution where $\alpha_c(p)= \frac{R-n_t
\mu_I(p)}{\sigma_I(p)}$. We obtain $\alpha_c'(p)= \frac{n_t \mu_I
\sigma_I'- n_t \mu_I'\sigma_I-R \sigma_I'}{\sigma_I^2}$ and \newline
$\alpha_c''(p)=\frac{(n_t \mu_I \sigma_I''-n_t \mu_I'' \sigma_I- R
\sigma_I'')\sigma_I^2-2\sigma_I\sigma_I'(n_t \mu_I \sigma_I'- n_t
\mu_I'\sigma_I-R \sigma_I')}{\sigma_I^4}$. We observe that, in the
equation $(\alpha_c'(p))^2 \alpha_c(p)-\alpha_c''(p)=0$, there are
terms in $n_t^3$, $n_t^2$, $n_t$ and constant terms w.r.t. $n_t$.
When $n_t$ becomes sufficiently large the first order terms can be
neglected, which implies that the solution is given by $\mu_I(p)=0$.
It can be shown that $\mu_{I}(0)=0$ and that $\mu_I$ is an
increasing function w.r.t. $p$ which implies that the unique
solution is $\tilde{p}_c=0$. Similarly to regime (a) we obtain the
trivial solution $p_c^*=0$.

\end{proof}

\section{Numerical results}
\label{sec:simus}

In this section, we present several simulations that illustrate our analytical results and verify the two conjectures stated. Since closed-form
expressions of the outage probability are not available in general, Monte Carlo simulations will be implemented. The exception is the MISO
channel for which the optimal energy-efficiency can be computed numerically (as we have seen in Sec. \ref{sec:sub_MISO}) without the need of
Monte Carlo simulations.


\emph{UPA, the quasi-concavity property and the large MIMO channels.}

Let us consider the case of UPA. In Fig. \ref{fig1}, we plot the GPR
$\Gamma_{\mathrm{UPA}}\left(p, R\right)$ as a function of the
transmit power $p \in [0,\ol{P}]$ W for an MIMO channel where
$n_r=n_t=n$ with $ n \in \{1,2,4,8\}$ and $\rho=10$ dB, $R = 1$
bpcu, $\ol{P}=1$ W. First, note that the energy-efficiency for UPA
is a quasi-concave function w.r.t. $p$, illustrating Conjecture
\ref{conjecture:MIMO_UPA}. Second, we observe that the optimal power
$p^*$ maximizing the energy-efficiency function is decreasing and
approaching zero as the number of antennas increases and also that
$\Gamma_{\mathrm{UPA}}\left(p^*,R\right)$ is increasing with $n$. In
Fig. \ref{fig2}, this dependence of the optimal energy-efficiency
and the number of antennas $n$ is depicted explicitly for the same
scenario. These observations are in accordance with the asymptotic
analysis in subsection \ref{subsec:asym_MIMO} for Regime (c).

Similar simulation results were obtained for the case where $n_t$  is fixed and $n_r$ is increasing, thus illustrating the asymptotic analysis
in subsection \ref{subsec:asym_MIMO} for Regime (a).

In Fig. \ref{fig3}, we plot the energy-efficiency $\Gamma_{\mathrm{UPA}}
\left(p, R \right)$ as a function of the transmit
power $p \in [0,\ol{P}]$ W for MIMO channel such that $n_r=2$, $n_t
\in \{1,2,4,8\}$ and $\rho=10$ dB, $R = 1$ bpcu, $\ol{P}=1$ W. The
difference w.r.t. the previous case, is that the optimal power $p^*$
does not go to zero when $n_t$ increases. This figure illustrates
the results obtained for Regime (b) in section
\ref{subsec:asym_MIMO} where the optimal power allocation $p^*_b
\rightarrow \frac{2^{\frac{R}{n_r}}-1}{\rho}=0.0414$~W and the
optimal energy-efficiency $\Gamma^*_{UPA} \rightarrow
\frac{\rho}{2(2^{\frac{R}{n_r}}-1)}= 12,07$ bit/Joule when $n_t
\rightarrow +\infty$.


\emph{UPA and the finite MISO channel}

In Fig. \ref{fig4}, we illustrate Proposition \ref{proposition:MISO}
for $n_t=4$. We trace the cases where the transmitter uses an
optimal UPA over only a subset of $\ell \in \{1,2,3,4\}$ antennas
for $\rho = 10$ dB, $R=3$ bpcu. We observe that: i) if $\ol{P}\leq
\frac{c}{c_1}$ then the beamforming PA is the generally optimal structure with $\mathrm{D}^*
= \ol{P} \ \mathrm{\textbf{Diag}}(\ul{e}_1)$; ii) if $\ol{P} \in
\left[\frac{c}{c_1} \frac{c}{c_2}\right)$ then using UPA over three
antennas is the generally optimal structure with $\mathrm{D}^* = \ol{P}/2 \
\mathrm{\textbf{Diag}}(\ul{e}_2)$; iii) if $\ol{P} \in
\left[\frac{c}{c_2} \frac{c}{c_3}\right)$ then using UPA over three
antennas is generally optimal with $\mathrm{D}^* =
 \ol{P}/3 \ \mathrm{\textbf{Diag}}(\ul{e}_3)$; iv) if $\ol{P}  \geq \frac{c}{c_4}$ then the UPA over all the antennas is optimal with
$\mathrm{D}^* =\frac{1}{4} \min\left\{\frac{4*c}{\nu_4}, \ol{P} \right\} \ \mb{I}_4$. The saturated regime illustrates the fact that it is not always
optimal to use all the available power after a certain threshold.


\emph{UPA and the finite MIMO channel}

Fig. \ref{fig5} represents the success probability, $1 -
P_{\mathrm{out}}(\mb{D},R)$, in function of the power constraint
$\ol{P}$ for $n_t=n_r=2$, $R=1$ bpcu, $\rho=3$ dB. Since the optimal
PA that maximizes the success probability is unknown (unlike the
MISO case) we use Monte-Carlo simulations and exhaustive search to
compare the optimal PA with the UPA and the beamforming PA. We
observe that the result is in accordance with Telatar's conjecture.
There exists a threshold $\delta= 0.16$ W such that if $\ol{P}\leq
\delta$, the beamforming PA is optimal and otherwise the UPA is
optimal. Of course, using all the available power is always optimal
when maximizing the success probability. The objective is to check
whether Conjecture \ref{conjecture:GPR_MIMO} is verified in this
particular case. To this purpose, Fig. \ref{fig6} represents the
energy-efficiency function for the same scenario. We observe that
for the exact threshold $\delta=0.16$ W, we obtain that if
$\ol{P}\leq \delta$ the beamforming PA using all the available power
is optimal. If $\ol{P} > \delta$ the UPA is optimal. Here, similarly
to the MISO case, we observe a saturated regime which means that
after a certain point it is not optimal w.r.t. energy-efficiency to
use up all the available transmit power. In conclusion, our
conjecture has been verified in this simulation.

Note that for the beamforming PA case we have explicit relations for
both the outage probability and the energy-efficiency (it is easy to
check that the MIMO with beamforming PA reduces to the SIMO case)
and thus Monte-Carlo simulations have not been used.

\section{Conclusion}
\label{sec:conclusion}

In this paper, we propose a definition of energy-efficiency metric which is the extension of the work in \cite{verdu-it-1990} to static MIMO
channels. Furthermore, our definition bridges the gap between the notion of capacity per unit cost \cite{verdu-it-1990} and the empirical
approach of \cite{goodman-pcom-2000} in the case of slow fading channels. In static and fast fading channels, the energy-efficiency is maximized
at low transmit power and the corresponding rates are also small. On the the other hand, the case of slow fading channel is not trivial and
exhibits several open problems. It is conjectured that solving the (still open) problem of outage minimization is sufficient to solve the
problem of determining energy-efficient precoding schemes. This conjecture is validated by several special cases such as the MISO case and
asymptotic cases. Many open problems are introduced by the proposed performance metric, here we just mention some of them:
\begin{itemize}
    \item First of all, the conjecture of the optimal precoding
    schemes for general MIMO channels needs to be proven.

    \item The quasi-concavity of the goodput-to-power ratio when
    uniform power allocation is assumed remains to be proven in the
    finite setting.

    \item A more general channel model should be considered. We have
    considered i.i.d. channel matrices but considering non zero-mean
    matrices with arbitrary correlation profiles appears to be a
    challenging problem for the goodput-to-power ratio.

    \item The connection between the proposed metric and the
    diversity-multiplexing tradeoff at high SNR has not been
    explored.

    \item Only single-user channels have been considered. Clearly,
    multi-user MIMO channels such as multiple access or interference
    channels should be considered.

    \item The case of distributed multi-user channels
    become more and more important for applications (unlicensed
    bands, decentralized cellular networks, etc.). Only one result
    is mentioned in this paper: the existence of a pure Nash
    equilibrium in distributed MIMO multiple access channels assuming
    uniform power allocation transmit policy.
\end{itemize}

%


\begin{figure}
  \begin{center}
  \includegraphics[angle=0,width=14cm]{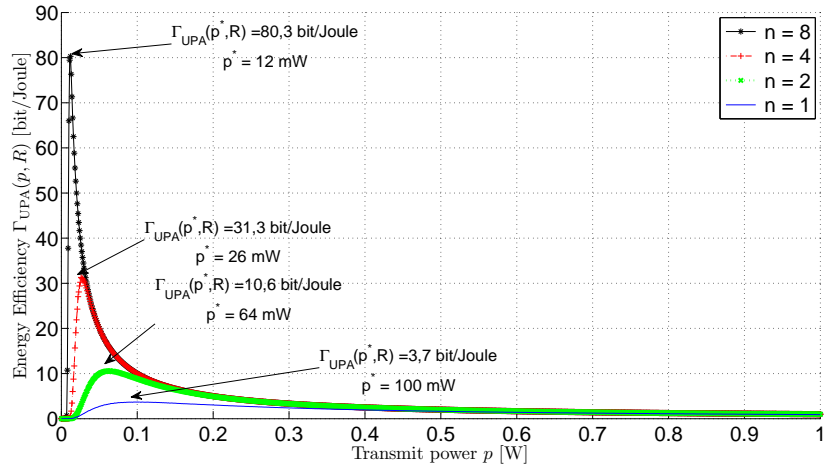}
    \caption{\scriptsize{Energy-efficiency (GPR) vs. transmit
    power $p \in [0,1]$ W for MIMO channels where
    $n_r=n_t=n \in
     \{1,2,4,8\}$, UPA $\mb{D}=\frac{p}{n_t}\mb{I}_{n_t}$,
      $\rho=10$ dB, $R = 1$ bpcu. Observe that the
       energy-efficiency is a quasi-concave function w.r.t. $p$. The
        optimal point $p^*$ is decreasing and $\Gamma_{\mathrm{UPA}}\left(p^*,R\right)$ is increasing with $n$.}}
\label{fig1}
\end{center}
\end{figure}


\begin{figure}
  \begin{center}
  \includegraphics[angle=0,width=14cm]{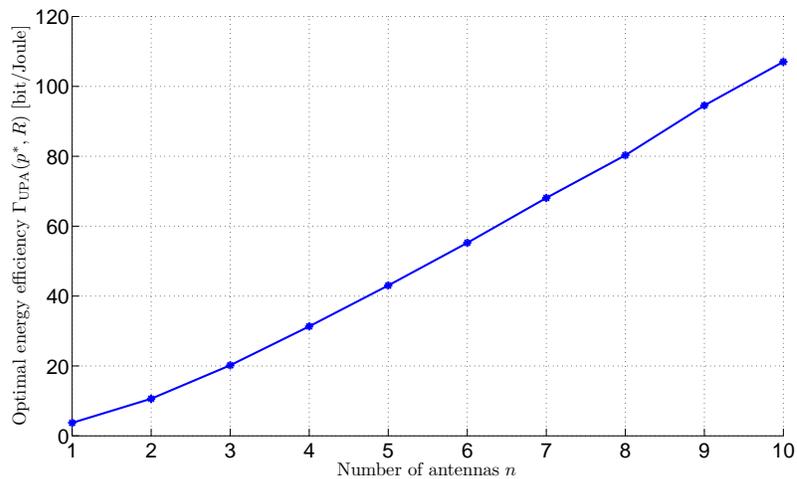}
    \caption{\scriptsize{Energy-efficiency vs. the number of
    antennas $n$ for MIMO $n_r=n_t=n \in \{1,2,4,8\}$, UPA,
    $\mb{D}=\frac{p}{n_t}\mb{I}_{n_t}$, $\rho=10$ dB, $R = 1$
    bpcu and $\ol{P}=1$ W. Observe that $\Gamma_{\mathrm{UPA}}\left(p^*,R\right)$ is increasing with $n$.}}
    \label{fig2}
\end{center}
\end{figure}

\begin{figure}
  \begin{center}
  \includegraphics[angle=0,width=14cm]{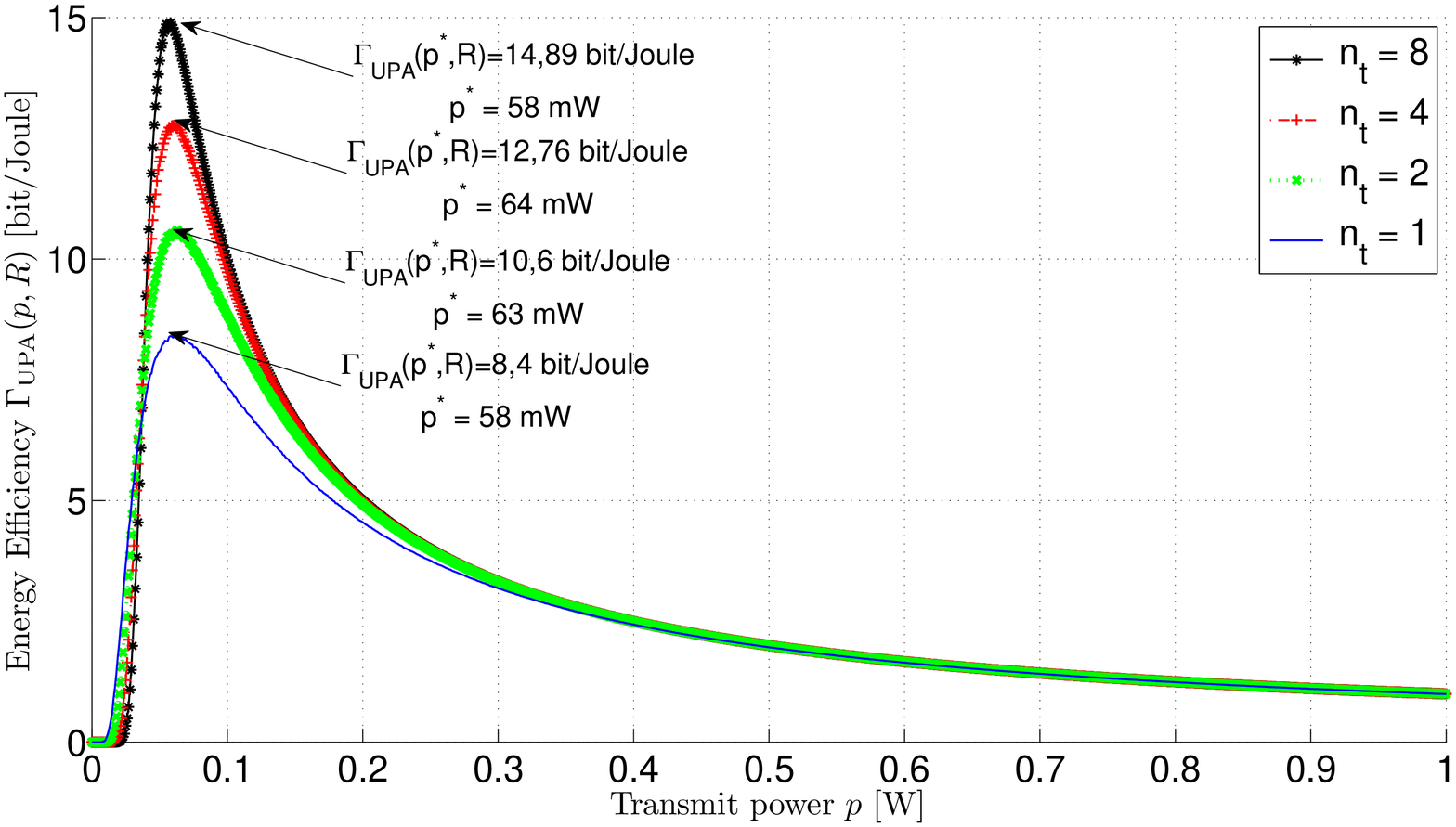}
    \caption{\scriptsize{Energy-efficiency vs. transmit power $p \in [0,1]$ W for MIMO $n_r=2$, $n_t\in \{1,2,4,8\}$, UPA $\mb{D}=\frac{p}{n_t}\mb{I}_{n_t}$, $\rho=10$ dB, $R = 1$ bpcu. Observe that the energy-efficiency is a quasi-concave function w.r.t. $p$. The optimal point $p^*$ is not decreasing with $n$ but almost constant.}}
    \label{fig3}
   \end{center}
\end{figure}



\begin{figure}
  \begin{center}
  \includegraphics[angle=0,width=14cm]{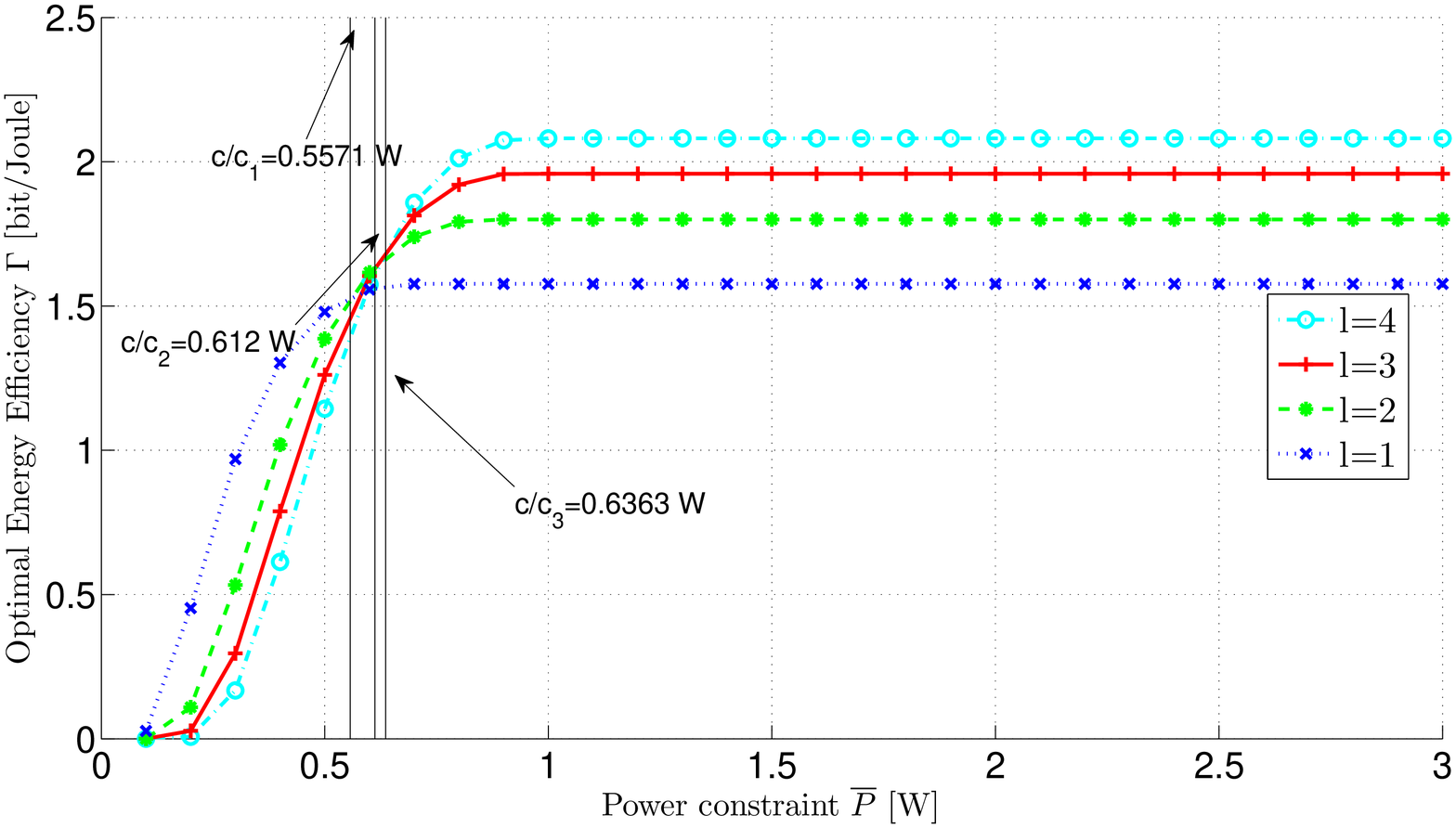}
    \caption{\scriptsize{Optimal energy-efficiency vs. constraint
     power for MISO $n_t=4$, $n_r=1$, UPA over a subset of $\ell \in \{1,2,3,4\}$ antennas, $\rho = 10$ dB, $R=3$ bpcu. We illustrate the results of Proposition \ref{proposition:MISO}. If $\ol{P}\leq \frac{c}{c_1}$ is low enough, the beamforming PA with full power is optimal. If $\ol{P} \geq \frac{c}{c_2}$ is high enough, the UPA is optimal but not with full power necessarily $\left(p^*=\min \{\frac{c}{\nu_4}, \ol{P}\} \right)$ which explains the saturated regime. }}
    \label{fig4}
  \end{center}
\end{figure}


\begin{figure}
  \begin{center}
  \includegraphics[angle=0,width=12cm]{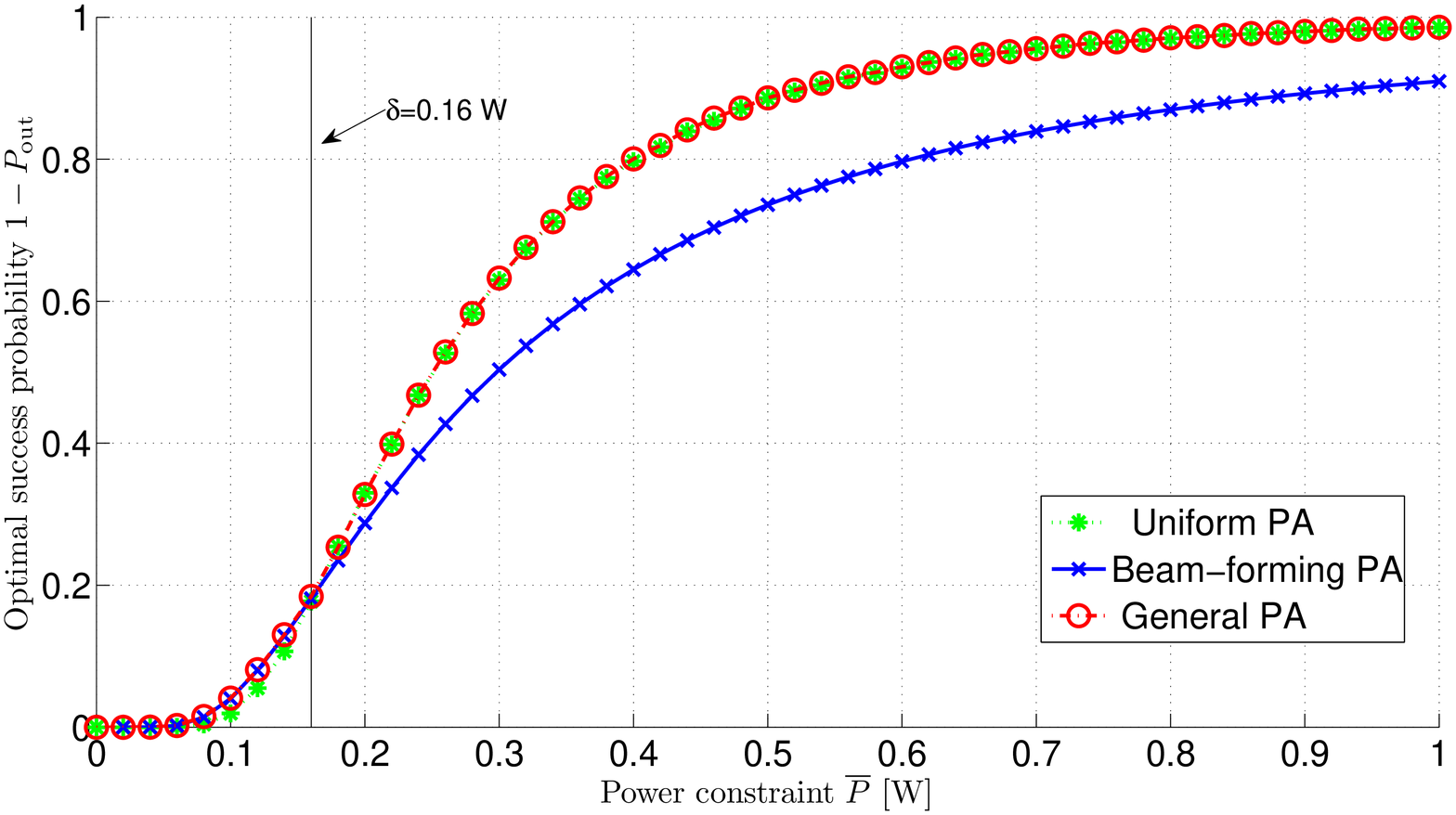}
    \caption{\scriptsize{Success probability vs. power constraint $\ol{P}$, comparison between beamforming PA, UPA and General PA for MIMO $n_t=n_r=2$, $R=1$ bpcu, $\rho=3$ dB. We observe that Telatar's conjecture is validated. There is a threshold, $\delta=0.16$ W, below which ($\ol{P}\leq \delta$) the beamforming PA is optimal and above it, UPA is optimal.}}
    \label{fig5}
    \end{center}
\end{figure}

\begin{figure}
  \begin{center}
  \includegraphics[angle=0,width=14cm]{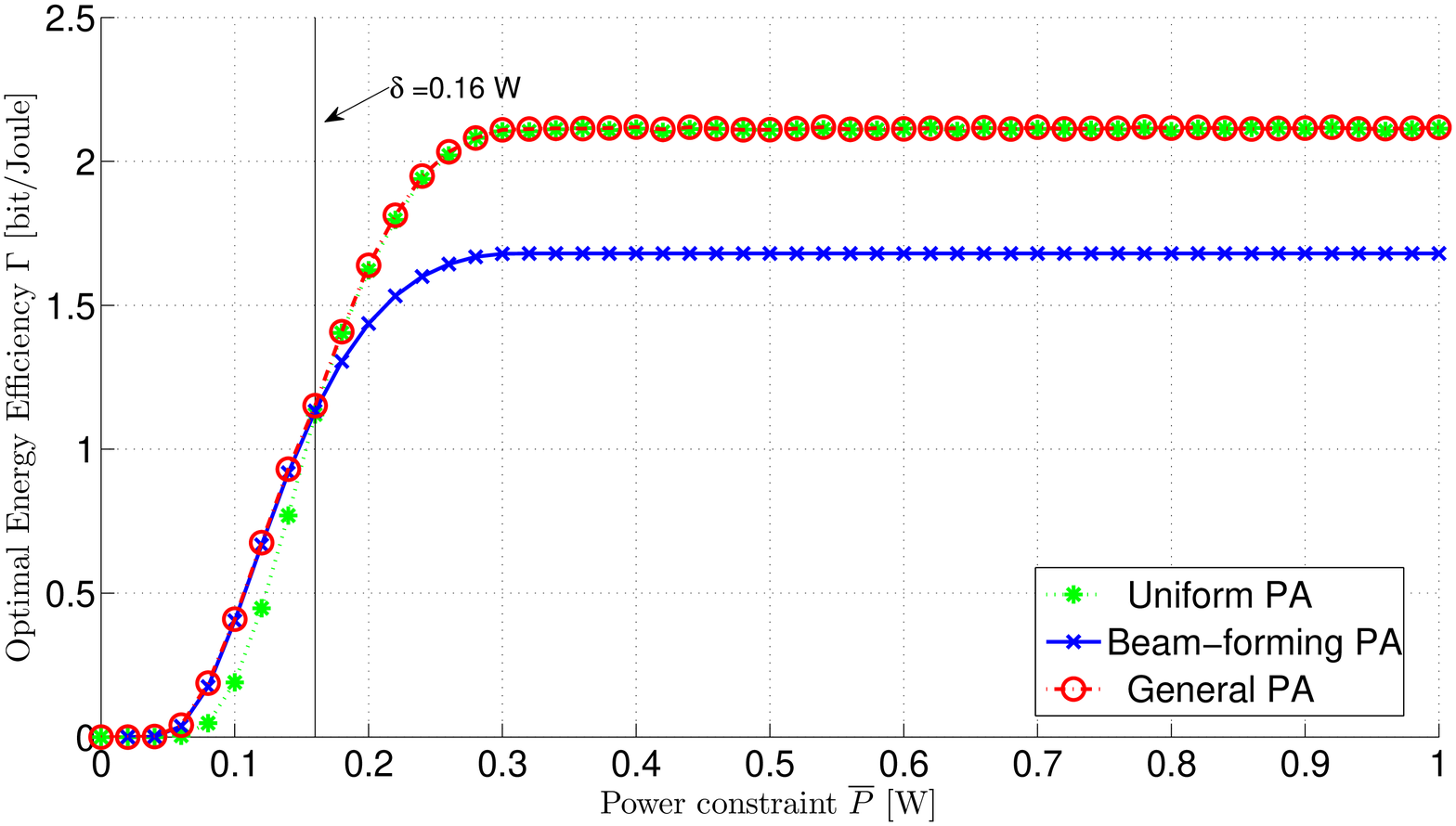}
    \caption{\scriptsize{Optimal energy-efficiency vs. power constraint $\ol{P}$, comparison between beamforming PA, UPA and General PA for MIMO $n_t=n_r=2$, $R=1$ bpcu, $\rho=3$ dB. We observe that our Conjecture \ref{conjecture:GPR_MIMO} is validated. For the exact same $\delta=0.16$ W, we have that for $\ol{P}\leq \delta$ the beamforming PA structure optimal and above it, UPA structure is optimal. }}
    \label{fig6}
  \end{center}
\end{figure}

\newpage

\appendices
\section{Proof of Proposition \ref{proposition:static-channels}}
\label{appendix:A}

 As $\mb{Q}$ is a positive semi-definite Hermitian matrix, it can always be spectrally decomposed as $\mb{Q}=\mb{U}\mb{D}\mb{U}^H$ where
$\mb{D}= \mathrm{\textbf{Diag}}(p_1,\hdots,p_{n_t})$ is a diagonal matrix representing a given PA policy and $\mb{U}$ a unitary matrix. Our goal
is to prove that, for every $\mb{U}$, $G_{\mathrm{static}}$ is maximized when $\mb{D} =\mathrm{\textbf{Diag}}(0, 0, ..., 0)$. To this end we
rewrite $G_{\mathrm{static}}$ as
\begin{equation}
\ds{G_{\mathrm{static}}(\mb{U} \ \mathrm{\textbf{Diag}}(p_1,\hdots,p_{n_t}) \ \mb{U}^H) = \frac{\log_2\left|\mb{I}_{n_r}+\ds{\sum_{i=1}^{n_t} p_i
\ul{g}_i \ul{g}_i^H }\right|}{\ds{\sum_{i=1}^{n_t}p_i}}},
\end{equation}
where $\ul{g}_i$ represents the $i^{th}$ column of the $n_r \times n_t$ matrix $\mb{G}=\sqrt{\rho}\mb{H}\mb{U}$ and proceed by induction on $n_t
\geq 1$.

First, we introduce an auxiliary quantity (whose role will be made
clear a little further)
\begin{equation}
\label{eq:tr_logdet} \begin{array}{lcl} E^{(n_t)}(p_1,\hdots,
p_{n_t})& \triangleq & \ds{
\mathrm{Tr}\left(\mb{I}_{n_r}+\sum_{i=1}^{n_t}p_i \ul{g}_i\ul{g}_i^H
\right)^{-1}\left(\sum_{i=1}^{n_t}p_i\ul{g}_i\ul{g}_i^H\right)}\\
& &\ds{- \log_2 \left|\mb{I}_{n_r}+ \sum_{i=1}^{n_r}p_i
\ul{g}_i\ul{g}_i^H\right| }.
\end{array}
\end{equation}
and prove by induction that it is negative that is,  $\forall
(p_1,\hdots, p_{n_t}) \in \mathbb{R}_+^{n_t}$,
$E^{(n_t)}(p_1,\hdots, p_{n_t}) \leq 0$.

For $n_t=1$, we have $E^{(1)}(p_1) = \mathrm{Tr}\left[(\mb{I}_{n_r}+p_1\ul{g}_1\ul{g}_1^H)^{-1} p_1\ul{g}_1\ul{g}_1^H\right] -
\log_2\left|\mb{I}_{n_r}+ p_1\ul{g}_1 \ul{g}_1^H\right|$. The first order derivative of $E^{(1)}(p_1)$ w.r.t. $p_1$ is:
\begin{equation}
\frac{\partial E^{(1)}}{\partial p_1}= -p_1
[\ul{g}_1^H(\mb{I}_{n_r}+p_1\ul{g}_1\ul{g}_1^H)^{-1}\ul{g}_1 ]^2 \leq 0
\end{equation}
and  thus $E^{(1)}(p_1) \leq E^{(1)}(0)=0$.

Now, we assume that $E^{(n_t-1)}(\ul{p}) \leq 0$ and want to prove
that $E^{(n_t)}(\ul{p}, p_{n_t}) \leq 0$, where $\ul{p}=(p_1,\hdots,
p_{n_t-1})$. It turns out that:
\begin{equation}
\frac{\partial E^{(n_t)}}{\partial p_{n_t}}= -\sum_{j=1}^{n_t} p_j \left|\ul{g}_j^H\left(\mb{I}_{n_r}+\sum_{i=1}^{n_t}p_i\ul{g}_i\ul{g}_i^H
\right)^{-1}\ul{g}_{n_t} \right|^2 \leq 0,
\end{equation}
and therefore $E^{(n_t)}(p_1,\hdots,p_{n_t-1},p_{n_t}) \leq
E^{(n_t)}(p_1,\hdots,p_{n_t-1},0)=E^{(n_t-1)}(p_1,\hdots,p_{n_t-1} )
\leq 0$.

As a second step of the proof, we want to prove by induction on $n_t
\geq 1$ that
\begin{equation}
\arg \max_{\ul{p}, p_{n_t}} G_{\mathrm{static}}^{(n_t)}(\ul{p},
p_{n_t})= \ul{0}.
\end{equation}
For $n_t=1$ we have
$G_{\mathrm{static}}^{(1)}(p_1)=\frac{\log_2|\mb{I}_{n_r}+ p_1\ul{g}_1
\ul{g}_1^H|}{p_1}=\frac{\log_2(1+p_1 \ul{g}_1^H\ul{g}_1)}{p_1}$
which reaches its maximum in $p_1=0$.

Now, we assume that $\arg \ds{\max_{\ul{p}}} \ G_{\mathrm{static}}^{(n_t-1)}(\ul{p})=\ul{0}$ and want to prove that $\ds{\arg
\max_{(\ul{p},p_{n_t})} G_{\mathrm{static}}^{(n_t)}(\ul{p},p_{n_t})=\ul{0}}$. \\Let $\ds{k =\arg \min_{i\in \{1, \hdots, n_t\}}
\mathrm{Tr}\left[\left(\mb{I}_{n_r}+\ds{\sum_{j=1}^{n_t}p_j \ul{g}_j\ul{g}_j^H} \right)^{-1}\ul{g}_i\ul{g}_i^H \right]}$. By calculating the first
order derivative of $G_{\mathrm{static}}^{(n_t)}$ w.r.t. $p_k$ one obtains that:
\begin{equation}
\frac{\partial G_{\mathrm{static}}^{(n_t)}}{\partial p_k}
 =  \ds{\frac{\mc{N}}{\left(\ds{\sum_{i=1}^{n_t}p_i}\right)^2}},
\end{equation}
with
\begin{equation}
\begin{array}{lcl}
\mc{N}&=&\left(\ds{\sum_{i=1}^{n_t}
p_i}\right)\mathrm{Tr}\left[\left(\mb{I}_{n_r}+\ds{ \sum_{j=1}^{n_t}p_j
\ul{g}_j\ul{g}_j^H}\right)^{-1}\ul{g}_k\ul{g}_k^H\right]
\\
& &-\log_2\left|\mb{I}_{n_r}+ \ds{\sum_{i=1}^{n_t}p_i
\ul{g}_i\ul{g}_i^H}\right|
\end{array}
\end{equation}
and thus $ \frac{\partial G_{\mathrm{static}}^{(n_t)}}{\partial p_k}
\leq
\ds{\frac{E^{(n_t)}(p_1,\hdots,p_{n_t})}{\left(\sum_{i=1}^{n_t}p_i
\right)^2} }\leq 0$ and $p_k^*=0$ for all $p_1, \hdots,
p_{k-1},p_{k+1}, \hdots, p_{n_t}$. We obtain that\\
$F^{(n_{t})}(p_1, \hdots, p_{k-1},0, p_{k+1}, \hdots, p_{n_t})$

$=F^{(n_t-1)}(p_1, \hdots, p_{k-1},p_{k+1}, \hdots, p_{n_t})$, which is maximized when $(p_1, \hdots, p_{k-1},p_{k+1}, \hdots, p_{n_t}) =
\ul{0}$ by assumption. We therefore have that $\mb{Q}^*=\mb{U}\mb{0}\mb{U}^H=\mb{0}$ is the solution that maximizes the function
$G_{\mathrm{static}}(\mb{Q})$. At last, to find the maximum reached by  $G_{\mathrm{static}}$ one just needs to consider the the equivalent of
the $\log_2\left|\mb{I}_{n_r}+ \rho \mb{H}\mb{Q}\mb{H}^H\right|$ around $\mb{Q}=\mb{0}$
\begin{equation}
\log_2\left|\mb{I}_{n_r}+ \rho \mb{H}\mb{Q}\mb{H}^H\right| \sim \frac{\rho}{n_t} \mathrm{Tr}(\mb{H}\mb{H}^H)
\end{equation}
and takes $\mb{Q} = \frac{q}{n_t} \mb{I}_{n_t}$ with $q \rightarrow 0$.

\section{Proof of Proposition \ref{proposition:slow_gen}}
\label{appendix:B}

The proof has two parts. First, we start by proving that if the
optimal
 solution is different than the uniform spatial power allocation
$\mb{P}^* \neq \frac{p}{n_t}\mb{I}_{n_t}$ with $p \in
\left[0,\ol{P}\right]$ then the solution is not trivial
$\mb{P}^* \neq \mb{0}$. We proceed by reductio ad absurdum. We
assume that the optimal solution is trivial $\mb{P}^* = \mb{0}$.
This
 means that when fixing $(p_2,\hdots,p_{n_t}) =
(0,\hdots, 0)$ the optimal $p_1 \in [0, \ol{P}]$ that maximizes the
energy-efficiency function is $p_1^*=0$. The energy-efficiency
function becomes:
\begin{equation}
\label{eq:simo} \Gamma(\mathrm{\textbf{Diag}}(p_1,0, \hdots,0),R) = R \frac{1- \mathrm{Pr}\left[\log_2 (1+ \rho p_1
\|\ul{h}_1\|^2)<R\right]}{p_1}
\end{equation}
where $\ul{h}_1$ represents the first column of the channel matrix
$\mb{H}$. Knowing that the elements in $\ul{h}_1$ are i.i.d. $h_{1j}
\sim \mc{C}\mc{N}(0,1)$ for all $j \in\{1,\hdots, n_r\}$ we have
that $|h_{1j}|^2 \sim \mathrm{expon}(1)$. The random variable
$\|\ul{h}_1\|^2= \ds{\sum_{j=1}^{n_r}}|h_{1j}|^2$ is the sum of
$n_r$ i.i.d. exponential random variables of parameter $\lambda=1$
and thus follows an $2n_r$ chi-square distribution (or an $n_r$
Erlang distribution) whose c.d.f. is known and given by
$\varsigma(x) = 1- \exp(-x)\ds{\sum_{k=0}^{n_r-1}} \frac{x^k}{k!}$.
We can explicitly calculate the outage probability and obtain the
energy-efficiency function:
\begin{equation}
\Gamma(\mathrm{\textbf{Diag}}(p_1,0, \hdots,0),R) = R \exp\left(-\frac{c}{p_1}\right) \sum_{k=0}^{n_r-1} \frac{c^k}{k!}\frac{1}{p_1^{k+1}}
\end{equation}
where $c =\frac{2^R-1}{\rho} > 0$. It is easy to check that
$\ds{\lim_{p_1 \rightarrow 0} \Gamma(p_1,R)=0}$, $\ds{\lim_{p_1
\rightarrow \infty} \Gamma(p_1,R)=0}$. By evaluating the first
derivative w.r.t. $p_1$, it is easy to check that the maximum is
achieved for $p_1^*=\frac{c}{\nu_{n_r}} \geq 0$ where $\nu_{n_r}$ is the unique
positive solution of the following equation (in $y$):
\begin{equation}\frac{1}{(n_r
-1)!} y^{ n_r} - \sum_{k=0}^{n_r -1} \frac{y^k}{k!}=0.
\end{equation}
Considering the power constraint the optimal transmission power is
$p_1^*=\min\{\frac{2^R-1}{\nu_{n_r} \rho}, \ol{P}\}$, which
contradicts the hypothesis and thus if the optimal solution is
different than the uniform spatial power allocation then the
solution is not trivial $\mb{P}^* \neq \mb{0}$.

\section{Proof Proposition \ref{proposition:MISO} }
\label{appendix:C}

Let $\ul{p}^T = (p_1,...,p_{n_t})$ be the vector of powers allocated to the different antennas $i\in\{1,...,n_t\}$ and thus $\mb{D}=
\mathrm{\textbf{Diag}}(\ul{p})$. Define the two sets: $\mc{C}(x) = \left\{\ul{p} \geq 0,\ds{ \sum_{i=1}^{n_t} } p_i \leq x \right\}$ and
$\Delta(x) = \left\{\ul{p} \geq 0, \ds{ \sum_{i=1}^{n_t} } p_i = x\right\}$. Using these notations, they key observation to be made is the
following:
\begin{equation}
\begin{array}{ccl}
\ds{\sup_{\ul{p} \in \mc{C}(\ol{P})} \Gamma^{\mathrm{MISO}}(\mb{D},R)} & \stackrel{(a)}{=} &  R \ds{ \sup_{\ul{p} \in \mc{C}(\ol{P})} \frac{1 -
P_{\mathrm{out}}^{\mathrm{MISO}}(\mb{D},
R)}{\ds{\sum_{i=1}^{n_t}} p_i}} \\
 & \stackrel{(b)}{=} & R \ds{\sup_{x \in [0, \ol{P}]}
\sup_{\ul{p} \in \Delta(x)} \frac{1 - P_{\mathrm{out}}^{\mathrm{MISO}}(\mb{D},
R)}{x}} \\
& \stackrel{(c)}{=} & R \ds{\sup_{x \in [0, \ol{P}]} \frac{g\left( \frac{c}{x}\right) }{x}}
\end{array}
\end{equation}
where $P_{\mathrm{out}}^{\mathrm{MISO}} = \mathrm{Pr}\left[\log\left( 1 + \rho \ds{\sum_{i=1}^{n_t} }p_i |h_i|^2\right) \leq R \right] $: (a) translates the
definition of the GPR; (b) follows from the property $\sup \{ A \cup B\}= \sup \{ \sup \{A\}, \sup \{ B \}\}$ for two sets $A$ and $B$, applied
to our context; in (c) the function $g(z) = \left\{g_{\ell}(z), \ \mathrm{if} z \in \left[\frac{c}{c_{\ell-1}}, \frac{c}{c_{\ell}}\right)\right.$ is a piecewise continuous function where $g_{\ell}(z) = 1 - \mathrm{Pr}\left[\frac{1}{\ell} \ds{\sum_{i=1}^{n_t}}|h_i|^2 \leq z \right]$ for $z \in
\left[\frac{c}{c_{\ell-1}}, \frac{c}{c_{\ell}} \right)$ and $\ell \in \{1, \hdots, n_t \}$. The function $g(z)$ corresponds to the solution of the minimization problem of the outage probability
\cite{jorswieck-ett-2007}.

Now, we study the function $g_{\ell}$. By calculating the first order
derivative of $\frac{1}{x}g_{\ell}\left(\frac{c}{x}\right)$ w.r.t. $x$ we obtain:
\begin{equation}
\frac{\mathrm{d}}{\mathrm{d} x}
\left\{\frac{1}{x}g_{\ell}\left(\frac{c}{x}\right)\right\} =
\frac{\mathrm{e}^{-\frac{\ell c}{x}}}{x^2}
\left[\frac{1}{(\ell-1)!}\left(\frac{\ell c}{x}\right)^{\ell} -
\ds{\sum_{j=0}^{\ell-1}
 \frac{1}{j!} \left(\frac{\ell c}{x}\right)^{j}} \right].
\end{equation}
Thus the function $\frac{1}{x} g \left(\frac{c}{x}\right)$ is
increasing for $x \in (0, x_{\ell})$ and decreasing on $x \in
(x_{\ell}, \infty)$. The maximum point is reached in $x_{\ell} =
\frac{\ell c}{y_{\ell}}$ where $y_{\ell}$ is the unique positive
solution of the equation $\phi_{\ell}(y)=0$ where
\begin{equation}
\phi_{\ell}(y) = \frac{1}{(\ell-1)!} y^{\ell} -
\sum_{i=0}^{\ell-1}\frac{1}{i!} y^i.
\end{equation}
We have that $\phi(0)= -1 < 0$ and
\begin{equation}
\begin{array}{lcl}
\phi_{\ell}(\ell) & = & \frac{1}{(\ell-1)!} \ell^{\ell}
-\ds{ \sum_{i=0}^{\ell-1}}\frac{1}{i!} \ell^i \\
& = & \ds{\sum_{i=0}^{\ell-1} }\frac{\ell-i-1}{i!} \ell^i \\
& > & 0.
\end{array}
\end{equation}
This implies that $y_{\ell} \leq \ell$ and thus $x_{\ell} \geq c $. Since $c_{n_t-1} \geq 1$ we also have $x_{\ell} \geq \frac{c}{c_{n_t-1}}$
for all $\ell \in \{1,\hdots,n_t-1\}$.

Therefore, all the functions $\frac{1}{x}g_{\ell}\left(\frac{c}{x}\right)$ are increasing on the intervals $\left(0,
\frac{c}{c_{n_t-1}}\right)$. Moreover, on the interval $ \left(\frac{c}{c_{n_t-1}},\infty\right)$, they are increasing on $
\left(\frac{c}{c_{n_t-1}}, x_{\ell} \right]$ and decreasing on $\left[x_{\ell}, \infty\right)$. Proposition \ref{proposition:MISO} follows
directly.

\section{Counter-example, TISO}
\label{appendix:D}

Consider the particular case where $n_t=2$ and $n_r=1$. From
Proposition \ref{proposition:MISO}, it follows that for a power
constraint $\ol{P} < \frac{c}{c_1}$ the beamforming power allocation
policy maximizes the energy-efficiency and
$\Gamma^{\mathrm{TISO}}(\mathrm{\textbf{Diag}}(\ol{P},0),R) =
\Gamma^{\mathrm{TISO}}(\mathrm{\textbf{Diag}}(0,\ol{P}),R) >
\Gamma^{\mathrm{TISO}}\left(\mathrm{\textbf{Diag}}\left(\frac{\ol{P}}{2},\frac{\ol{P}}{2}\right),R
\right)$ . The function
$\Gamma^{\mathrm{TISO}}(\mathrm{\textbf{Diag}}(p_1,p_2),R)$ with
$(p_1,p_2) \in \mc{P}_2 \triangleq \{(p_1,p_2)\in \mathbb{R}_+^2 \ |
\ p_1+p_2 \leq \ol{P} \}$ denotes the energy-efficiency function. We
want to prove that
$\Gamma^{\mathrm{TISO}}(\mathrm{\textbf{Diag}}(p_1,p_2),R)$ is not
quasi-concave w.r.t. $(p_1,p_2) \in \mc{P}_2$. This amounts to
finding a level $\gamma \geq 0$ such that the corresponding
upper-level set $\mc{U}_{\gamma} = \left\{(p_1,p_2)\in \mc{P}_2 \ |
\ \Gamma^{\mathrm{TISO}}(\mathrm{\textbf{Diag}}(p_1,p_2),R) \geq
\gamma\right\}$ is not a convex set (see \cite{boyd-book-2004} for a
detailed analysis on quasi-concave functions). Consider an arbitrary
$0< q < \min\left\{ \ol{P}, \frac{c}{c_1} \right\}$ such that
$\Gamma^{\mathrm{TISO}}(\mathrm{\textbf{Diag}}(q,0),R) =
\Gamma^{\mathrm{TISO}}(\mathrm{\textbf{Diag}}(0,q),R) <
\Gamma^{\mathrm{TISO}}\left(\mathrm{\textbf{Diag}}\left(\frac{q}{2},\frac{q}{2}\right),R
\right)$. It turns out that all upper-level sets $\mc{U}_{\gamma_q}$
with $\gamma_q \triangleq
\Gamma^{\mathrm{TISO}}(\mathrm{\textbf{Diag}}(q,0),R)$ are not
convex sets. This follows directly from the fact that $(q,0), (0,q)
\in \mc{U}_{\gamma_q}$ but $\left(\frac{q}{2},\frac{q}{2}\right)
\notin \mc{U}_{\gamma_q}$ since
$\Gamma^{\mathrm{TISO}}\left(\mathrm{\textbf{Diag}}\left(\frac{q}{2},\frac{q}{2}\right),R
\right) < \gamma_q$.

\section{Extreme SNR cases, GPR}
\label{appendix:E}

In \cite{jorswieck-ett-2007}, the authors proved that in the low SNR regime the outage probability $P_{\mathrm{out}}(\ul{p}, R)$ is
Schur-concave w.r.t. $\ul{p}$. This means that for any vectors $\ul{p}$, $\ul{q}$ such that $\ul{p} \succ \ul{q}$ then $P_{\mathrm{out}}(\ul{p},
R) \leq P_{\mathrm{out}}(\ul{q}, R)$. The operator $\succ$ denotes the majorization operator which will be briefly described (see
\cite{marshall-book-1979} for details). For any two vectors $\ul{p}, \ul{q} \in \mathbb{R}_+^{n_t}$, $\ul{p}$ majorizes $\ul{q}$ (denoted by
$\ul{p} \succ \ul{q}$) if $\ds{\sum_{k=1}^{m}} p_k \geq \ds{\sum_{k=1}^{m}} q_k$, for all $m \in \{1,\hdots, n_t-1\}$ and $\ds{\sum_{k=1}^{n_t}}
p_k = \ds{\sum_{k=1}^{n_t}} q_k$. This operator induces only a partial ordering. The Schur-convexity and $\prec$ operator can be defined in an
analogous way. Also, an important observation to be made is that the beamforming vector majorizes any other vector, whereas the uniform vector
is majorized by any other vector (provided the sum of all elements of the vectors is equal). Otherwise stated, $x \ul{e}_1  \succ \ul{p} \succ
\frac{x}{n_t} \ul{\mb{1}} $ for any vector $\ul{p}$ such that $\ds{\sum_{i=1}^{n_t} }p_i = x$ and $\ul{\mb{1}} =(1,1,\hdots,1)$ and $\ul{e}_1
\in \mc{S}_1$.

It is straightforward to see that if $P_{\mathrm{out}}(\textbf{Diag}(\ul{p}), R)$ is Schur-concave
w.r.t. $\ul{p}$ then $1 - P_{\mathrm{out}}(\textbf{Diag}(\ul{p}), R)$ is Schur-convex w.r.t. $\ul{p}$.
Since the majorization operator implies the
sum of all elements of the ordered vectors to be identical, $\Gamma(\textbf{Diag}(\ul{p}), R)=\frac{1 - P_{\mathrm{out}}(\textbf{Diag}(\ul{p}), R)}{\ds{\sum_{i=1}^{n_t}p_i}}$
will also be Schur-convex w.r.t. $\ul{p}$ and thus is maximized by a beamforming vector. Using the same notations as in Appendix \ref{appendix:C} we obtain:

\begin{equation}
\begin{array}{ccl}
\ds{\sup_{\ul{p} \in \mc{C}(\ol{P})} \Gamma(\textbf{Diag}(\ul{p}),R)}
 & = & \ds{\sup_{x \in [0, \ol{P}]}
\frac{1}{x} \sup_{\ul{p} \in \Delta(x)} [1 -
P_{\mathrm{out}}(\textbf{Diag}(\ul{p}),
R)]} \\
& \stackrel{(a)}{=} & \ds{\sup_{x \in [0, \ol{P}]}\frac{1}{x} [1 -
\mathrm{Pr}[\log(1+ x \rho \ul{h}_1^H \ul{h}_1) \leq R ],} \\

& = & \ds{\sup_{x \in [0, \ol{P}]}\frac{1}{x} \left\{1 -
Pr\left[\frac{1}{n_r} \sum_{j=1}^{n_r} |h_{1j}|^2  \leq \frac{c}{n_r x} \right] \right\},
} \\
& \stackrel{(b)}{=} &  \ds{\sup_{x \in [0, \ol{P}]}\frac{g_{n_r}\left(\frac{c}{n_r x}\right)}{x} },
\end{array}
\end{equation}

where (a) follows by considering beamforming power allocation policy
on the first transmit antenna (with no generality loss) and replacing $\ul{p} = x\ul{e}_1$
with $\ul{e}_1 = (1, 0, \hdots, 0)$ and $\ul{h}_1$ denoting the
first column of the channel matrix; in (c) we make use the
definition in Appendix \ref{appendix:C} for the function
$\frac{1}{x} g_{n_r}\left(\frac{c}{n_r x}\right)$ which has a unique
optimal point in $\min \left\{ \frac{c}{y_{n_r}} , \ol{P} \right\}$,
with $y_{n_r}$ the unique solution of $\Phi_{n_r}(y) = 0$. Since
$\sigma^2 \rightarrow 0$ then $c \rightarrow + \infty$ and thus the
optimal power allocation is $\ul{p}^* = \ol{P} \ul{e}_1$.

Similarly, for the high SNR case we have:

\begin{equation}
\begin{array}{ccl}
\ds{\sup_{\ul{p} \in \mc{C}(\ol{P})} \Gamma(\textbf{Diag}(\ul{p}),R)}
 & = & \ds{\sup_{x \in [0, \ol{P}]}
\frac{1}{x} \sup_{\ul{p} \in \Delta(x)} [1 -
P_{\mathrm{out}}(\textbf{Diag}(\ul{p}),
R)]} \\
& = & \ds{\sup_{x \in [0, \ol{P}]}\frac{1}{x} \left[1 -
P_{\mathrm{out}}\left(\textbf{Diag}\left(\frac{x}{n_t} (1,\hdots,1)\right), R\right)\right] }.
\end{array}
\end{equation}

We have used the results in \cite{jorswieck-ett-2007}, where the UPA was proven to minimize the outage probability.

Let us now consider the limit of the energy-efficiency function when
$p \rightarrow 0$, $\sigma^2 \rightarrow 0$ such that $\frac{ p}{\sigma^2}
\rightarrow \xi$ with $\xi$ a positive finite constant. We obtain
that $1- P_{\mathrm{out}}\left(\frac{x}{n_t}\mb{I}_{n_t},R\right)
\rightarrow \mathrm{Pr}\left[\left|\mb{I}_{n_r} + \frac{\xi}{n_t}
\mb{H}\mb{H}^H \right|\right] >0$ which implies directly that
$\Gamma\left(\frac{x}{n_t}\mb{I}_{n_t},R\right) \rightarrow + \infty$.


\bibliography{biblio}

\end{document}